\documentclass[letterpaper,11pt]{article}

\usepackage[arXiv]{optional}

\usepackage{verbatim}

\usepackage[margin=1in]{geometry}
\usepackage{setspace}
\usepackage{graphicx}
\usepackage{epsfig}

\usepackage{microtype} %
\usepackage{array} %

\usepackage{amssymb}
\usepackage{amsmath}
\usepackage{amsthm}
\usepackage{mathtools} %
\usepackage{stmaryrd} %
\usepackage{stackrel} %
\usepackage{latexsym} %

\usepackage{tikz}
\usetikzlibrary{arrows}
\usetikzlibrary{arrows.meta}
\usetikzlibrary{decorations.pathreplacing}
\usetikzlibrary{cd}
\usepackage[all]{xy} %

\usepackage{algorithm}
\usepackage[noend]{algpseudocode}
\algnewcommand{\LineComment}[1]{\Statex\hspace{\algorithmicindent}\(\triangleright\) #1}
\algnewcommand\algorithmicforeach{\textbf{for each}}
\algdef{S}[FOR]{ForEach}[1]{\algorithmicforeach\ #1\ \algorithmicdo}
\usepackage{etoolbox} %

\usepackage{caption}
\usepackage{subcaption}

\usepackage{makecell} %
\usepackage{booktabs} %
\usepackage{multirow} %
\usepackage{tablefootnote} %

\usepackage[hidelinks]{hyperref} %
\usepackage[symbol]{footmisc}
\usepackage{makeidx} %
\usepackage{url} %
\usepackage{color}
\usepackage{xcolor}
\usepackage{xspace} %
\usepackage[normalem]{ulem} %
\usepackage{soul} %
\usepackage{ifpdf} %
\usepackage{extarrows}
\usepackage{stackengine}    
\usepackage{lipsum}
\usepackage{bm}
\usepackage{cleveref} %
\usepackage{cite} %

\usepackage{ifthen} %

\makeatletter
\expandafter\patchcmd\csname\string\algorithmic\endcsname{\itemsep\z@}{\itemsep=0.25ex}{}{}
\makeatother

\theoremstyle{plain}
    \newtheorem{theorem}{Theorem}
    \newtheorem{proposition}[theorem]{Proposition}

    \newtheorem{observation}[theorem]{Observation}
    
    \newtheorem*{algr*}{Algorithm}
\theoremstyle{definition}

    \newtheorem{remark}[theorem]{Remark}
    \newtheorem*{remark*}{Remark}
    \newtheorem*{example*}{Example}
    
\usepackage{enumitem}
\setlist[itemize]{leftmargin=\parindent}
\setlist[enumerate]{leftmargin=\parindent}
\setlist[description]{font=\bfseries,leftmargin=\parindent}

\title{Updating zigzag representatives efficiently}

\author{Tamal K. Dey\thanks{Department of Computer Science, Purdue University, West Lafayette, IN, USA. \texttt{tamaldey@purdue.edu}}
\and Tao Hou\thanks{Department of Computer Science, University of Oregon, Eugene, OR, USA. \texttt{taohou@uoregon.edu}}
\and  Dmitriy Morozov\thanks{Lawrence Berkeley National Laboratory, Berkeley, CA, USA. \texttt{dmitriy@mrzv.org}}
}

\date{}

\crefname{equation}{Eq.}{Eqs.}
\crefname{figure}{Fig.}{Figs.}

\newcommand{\para}[1]{\opt{arXiv,siam}{\paragraph{#1}}\opt{SoCG}{\subparagraph*{#1}}}

\let\bar\overline

\let\emptyset\varnothing

\newcommand{\Hm}{H}
\newcommand{\Chn}{C}
\newcommand{\Zyc}{Z}

\newcommand{\field}{\Bbbk}

\newcommand{\Pers}{\mathsf{Pers}}
\newcommand{\LPers}{\mathsf{LPers}}

\renewcommand{\bar}[1]{\overline{#1}}

\newcommand{\lbarrowspace}{\;}

\let\leftrightarrowsp\lrarrowsp

\newcommand{\incto}{\hookrightarrow}
\newcommand{\inctosp}[1]{\xhookrightarrow{\lbarrowspace#1\lbarrowspace}}
\newcommand{\bakincto}{\hookleftarrow}
\newcommand{\bakinctosp}[1]{\xhookleftarrow{\lbarrowspace#1\lbarrowspace}}

\newcommand{\given}{\,|\,}
\newcommand{\Set}[1]{\{#1\}}

\newcommand{\Fcal}{\mathcal{F}}

\newcommand{\Ibb}{\mathbb{I}}

\newcommand{\Zbb}{\mathbb{Z}}

\newcommand{\aG}{\alpha}

\newcommand{\DG}{\Delta}

\newcommand{\gG}{\gamma}
\newcommand{\GG}{\Gamma}

\newcommand{\oG}{\omega}

\newcommand{\sG}{\sigma}

\newcommand{\tG}{\tau}

\newcommand{\Dim}{p}
\newcommand{\dimx}{q}
\newcommand{\birth}{b}

\newcommand{\death}{d}
\newcommand{\fcnt}{m}
\newcommand{\scnt}[1]{\ifthenelse{\equal{#1}{}}{n}{n_{#1}}}
\newcommand{\scntx}[1]{\ifthenelse{\equal{#1}{}}{n^*}{n^*_{#1}}}

\newcommand{\ssx}{\sigma}
\newcommand{\tsx}{\tau}
\newcommand{\xssx}{\kappa}
\newcommand{\xssxx}{\nu}
\newcommand{\filt}{\Fcal}
\newcommand{\cplx}{K}
\newcommand{\dfilt}{\bar{\Fcal}}
\newcommand{\dcplx}{\bar{K}}
\newcommand{\dssx}{\bar{\sigma}}

\newcommand{\ucplx}{L}
\newcommand{\ussx}{\tau}
\newcommand{\filtnz}{{\Fcal_\mathrm{nz}}}
\newcommand{\filtx}{\Fcal^*}
\newcommand{\overbar}[1]{\overline{#1\mkern-6mu}\mkern 6mu}
\newcommand{\dfiltx}{\overbar{\Fcal^*}}
\newcommand{\dcplxx}{\overbar{K^*_i}}
\newcommand{\filtxnz}{\Fcal^*_\mathrm{nz}}
\let\cdec\hat %

\newcommand{\cplxx}{K^*}

\newcommand{\matel}[3]{{#1}[#2,#3]}
\newcommand{\col}[2]{{#1}[\cdot,#2]}
\newcommand{\row}[2]{{#1}[#2,\cdot]}
\newcommand{\Dmat}{D}
\newcommand{\Rmat}{R}
\newcommand{\Vmat}{V}
\newcommand{\Dmatx}{D^*}
\newcommand{\Rmatx}{R^*}
\newcommand{\Vmatx}{V^*}
\newcommand{\Smat}{T}
\newcommand{\Tmat}{S}
\newcommand{\tran}{\mathsf{T}}
\newcommand{\pivot}{\mathrm{piv}}

\makeatletter
\newcommand*{\da@rightarrow}{\mathchar"0\hexnumber@\symAMSa 4B }
\newcommand*{\da@leftarrow}{\mathchar"0\hexnumber@\symAMSa 4C }
\newcommand*{\xdashrightarrow}[2][]{%
  \mathrel{%
    \mathpalette{\da@xarrow{#1}{#2}{}\da@rightarrow{\;}{}}{}%
  }%
}
\newcommand{\xdashleftarrow}[2][]{%
  \mathrel{%
    \mathpalette{\da@xarrow{#1}{#2}\da@leftarrow{}{}{\;}}{}%
  }%
}
\newcommand{\xdashleftrightarrow}[2][]{%
  \mathrel{%
    \mathpalette{\da@xarrow{#1}{#2}\da@leftarrow\da@rightarrow{}{}}{}%
  }%
}
\newcommand*{\da@xarrow}[7]{%
  \sbox0{$\ifx#7\scriptstyle\scriptscriptstyle\else\scriptstyle\fi#5#1#6\m@th$}%
  \sbox2{$\ifx#7\scriptstyle\scriptscriptstyle\else\scriptstyle\fi#5#2#6\m@th$}%
  \sbox4{$#7\dabar@\m@th$}%
  \dimen@=\wd0 %
  \ifdim\wd2 >\dimen@
    \dimen@=\wd2 %
  \fi
  \count@=2 %
  \def\da@bars{\dabar@\dabar@}%
  \@whiledim\count@\wd4<\dimen@\do{%
    \advance\count@\@ne
    \expandafter\def\expandafter\da@bars\expandafter{%
      \da@bars
      \dabar@ 
    }%
  }%
  \mathrel{#3}%
  \mathrel{%
    \mathop{\da@bars}\limits
    \ifx\\#1\\%
    \else
      _{\copy0}%
    \fi
    \ifx\\#2\\%
    \else
      ^{\copy2}%
    \fi
  }%
  \mathrel{#4}%
  \!\!
}
\makeatother

\newcounter{desccounter}

\opt{arXiv}{
\begin{document}}

\maketitle

\begin{abstract}
Computation of zigzag persistence has progressed in recent years, 
with results showing that complexities of many problems 
closely align with those in the non-zigzag setting.
The major efficiency gap now lies in the updating of zigzag representatives.
In this paper, we propose efficient algorithms for updating zigzag representatives
based on a recent algorithm for extracting zigzag representatives from
a $R=DV$ decomposition of a constructed non-zigzag.
The main difficulty for designing our update algorithms lies in the adjacency change
occurring in two operations that lengthen or shorten a filtration.
Despite the adjacency change, we find that the update can still be
done efficiently in quadratic time.

\end{abstract}

\section{Introduction\opt{siam}{.}}

In persistent homology, a barcode of a filtered topological space provides a quantitative summary
of its homological features.
While the bars in the barcode effectively identify the duration of 
significant homological features, 
the underlying geometry is not captured by these durations alone.
\emph{Cycle representatives} add the missing information by locating the specific features in the topological spaces as they are born and die.
Designing algorithms
to find representatives for the ordinary (non-zigzag) persistence has been
fruitful because
an ordinary representative consists of a \emph{single} cycle at birth that is included
into all subsequent spaces~\cite{cohen2006vines,MS25}.
As a result, a representative takes $O(m)$ space to store
and summing two representatives takes $O(m)$ time,
where $m$ is the number of insertions in the filtration.
Since most algorithms for computing ordinary representatives
require adding representatives,
the compact representation makes it possible to design
efficient algorithms for computing the ordinary representatives.

In zigzag persistence, since simplices can
enter and leave a filtration, a representative cycle at one step 
may not exist at a later time. 
Hence,
representatives \emph{compatible} with the zigzag,
that is, 
ones diagonalizing the linear maps,
take $O(mn)$ space each,
where $m$ is the number of insertions and deletions in the filtration
and $n$ is the maximum size of any complex.
As as result,
summing two such representatives 
takes $O(mn)$ time.
Perhaps this is why the time complexity for \emph{updating} zigzag representatives
still could not catch up with the ordinary case,
considering that almost all other computational results for zigzags
match the complexity for the ordinary case~\cite{DHM24,DBLP:conf/compgeom/DeyH24,DBLP:conf/compgeom/Dey0M25}.
Specifically, the work~\cite{DBLP:conf/compgeom/DeyH24} proposes algorithms
for updating zigzag \emph{barcodes} over eight elementary operations allowing
a filtration to be transformed into any other one.
Time complexities for all the operations in~\cite{DBLP:conf/compgeom/DeyH24} match the ordinary case,
e.g., it is shown that elongating or shortening a zigzag filtration can be done in $O(m^2)$ time.
However, the approach does not directly generalize
to an efficient update of zigzag {representatives},
because the representatives 
maintained for the constructed {up-down} filtrations do not directly translate
to representatives for the original zigzag.

In this paper, we provide efficient update algorithms for zigzag representatives 
whose complexity matches the ordinary case, with the only extra cost
to output the representatives, a cost that cannot be avoided.
Our update algorithms are built on a recent work~\cite{DBLP:conf/compgeom/Dey0M25}
on extracting zigzag representatives efficiently from 
the $R=DV$ decomposition~\cite{cohen2006vines} of a constructed non-zigzag 
filtration.
In a nutshell, the authors~\cite{DBLP:conf/compgeom/Dey0M25} show that 
a zigzag representative in homological dimension $p$ can be extracted 
from a column of the matrix $R$ or $V$ in $O(p\cdot m\log m+C)$ time, 
where $C$ is the size of the representative.
This is almost as efficient as it can get because in most settings
$C$ dominates the $p\cdot m\log m$ term.
Relying on these results,
we provide efficient update algorithms for 
the $R=DV$ decomposition constructed for the ordinary filtration derived from the zigzag. This allows us to extract the representatives using
methods in~\cite{DBLP:conf/compgeom/Dey0M25}.
Time complexities for updating $R=DV$ decomposition
over the eight operations match those in the non-zigzag 
case: operations that {transpose}
two inclusions run in $O(m)$ time 
and operations that lengthen or shorten a filtration
run in 
$O(m^2)$ time.

Since the update of zigzag representatives
eventually boils down to updating the ordinary representatives in $R=DV$,
why not just use the existing update algorithms~\cite{cohen2006vines,giunti2023pruning}?
We find that, among the eight operations, six of them can be easily updated
using existing
algorithms~\cite{cohen2006vines}.
However, the remaining two operations, specifically, the outward expansions
and contractions (see \Cref{sec:update-oper}),
cannot be implemented using existing approaches because of 
an \emph{adjacency change} on the simplices\footnote{More accurately,
the adjacency change occurs on \emph{$\DG$-cells}
in some \emph{$\DG$-complexes} constructed; see \Cref{sec:back} for details.}
(see \Cref{sec:bound-chg}).
In particular, the transposition algorithm~\cite{cohen2006vines}
and the removal algorithm~\cite{giunti2023pruning}
assume that boundaries of simplices stay the same before
and after the update, 
which is not the case in the outward expansions
and contractions.
Hence, we mainly focus on addressing the update
of $R=DV$ for the two operations causing adjacency changes;
see \Cref{sec:out-contra,sec:out-exp}.

\section{Background\opt{siam}{.}}\label{sec:back}
\begin{figure}[t]
  \centering
  \captionsetup[subfigure]{justification=centering}

  \begin{subfigure}[t]{0.08\textwidth}
  \centering
  \includegraphics[width=\linewidth]{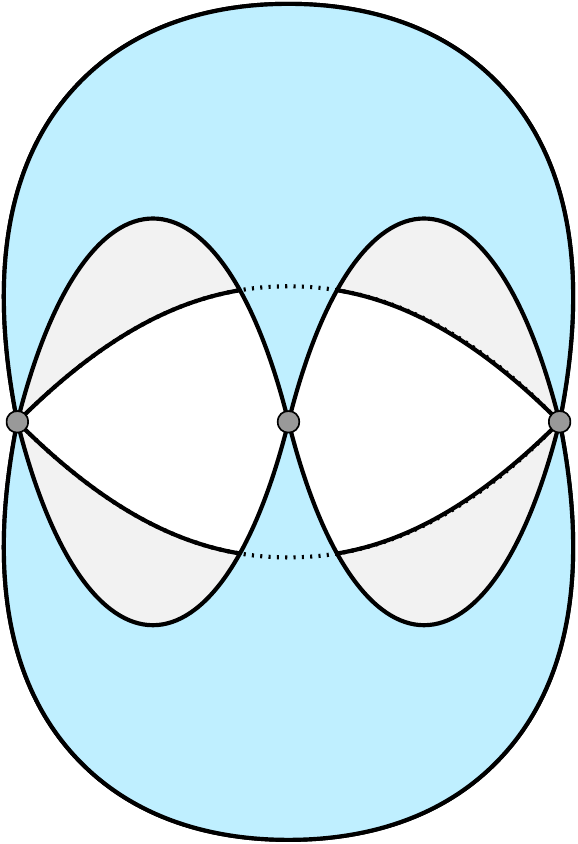}
  \label{fig:two-tri-0}
  \end{subfigure}
  \hspace{4em}
  \begin{subfigure}[t]{0.08\textwidth}
  \centering
  \includegraphics[width=\linewidth]{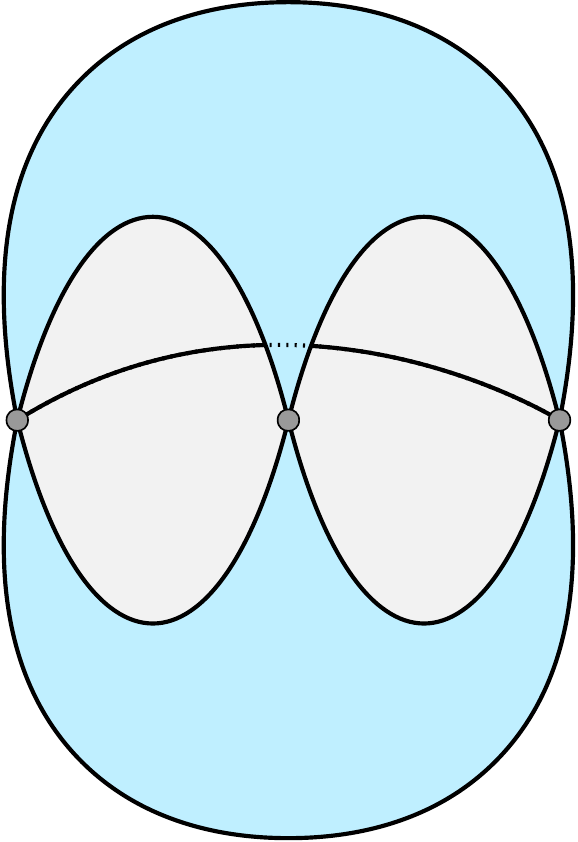}
  \label{fig:two-tri-1}
  \end{subfigure}
  \hspace{4em}
  \begin{subfigure}[t]{0.08\textwidth}
  \centering
  \includegraphics[width=\linewidth]{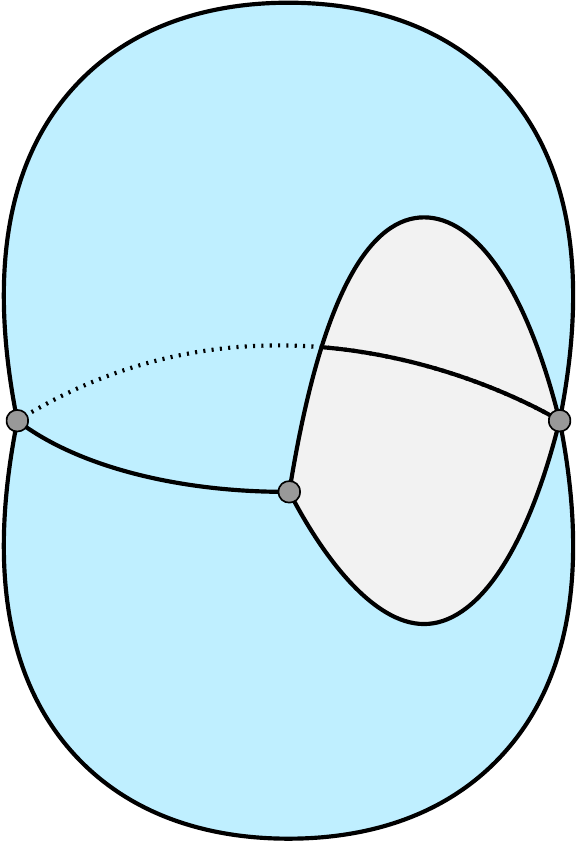}
  \label{fig:two-tri-2}
  \end{subfigure}
  \hspace{4em}
  \begin{subfigure}[t]{0.08\textwidth}
  \centering
  \includegraphics[width=\linewidth]{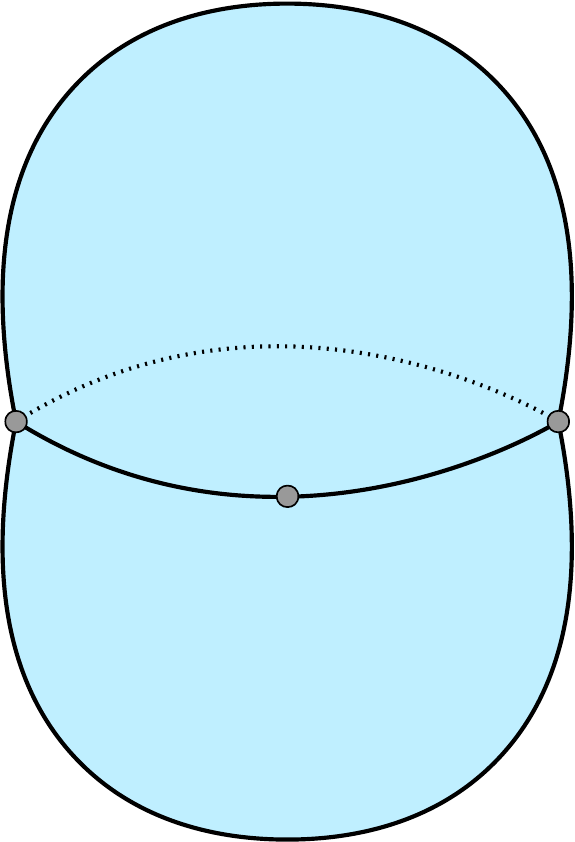}
  \label{fig:two-tri-3}
  \end{subfigure}

  \caption{Examples~\cite{DBLP:conf/esa/DeyH22} of $\DG$-complexes with two triangles {having the same set of vertices} sharing 0, 1, 2, or 3 edges on their boundaries.}
  \label{fig:bubble}
\end{figure}

\para{Zigzag persistence and representative.}
While the inputs all consist of simplicial complexes,
algorithms in this paper internally work on 
\emph{$\DG$-complexes}~\cite{79945ceb-9a4c-3c90-a4ac-148c36996187,hatcher2002algebraic}
by using of the conversion in~\cite{DBLP:conf/esa/DeyH22,DBLP:conf/compgeom/Dey0M25}.
Building blocks of $\DG$-complexes are
called \emph{cells} or \emph{$\DG$-cells},
which are similar to simplices
but could share common faces in more
relaxed ways; see \Cref{fig:bubble}
for examples
and also~\cite{79945ceb-9a4c-3c90-a4ac-148c36996187,hatcher2002algebraic} for a formal definition.
A {\it zigzag filtration} (or simply {\it filtration})
is  a sequence of $\DG$-complexes
\begin{equation*}
\Fcal: K_0 \leftrightarrow K_1 \leftrightarrow 
\cdots \leftrightarrow K_\fcnt,
\end{equation*}
in which each
$K_i\leftrightarrow K_{i+1}$ is either a forward inclusion $K_i\incto K_{i+1}$
(an insertion of cells)
or a backward inclusion $K_i\bakincto K_{i+1}$
(a deletion of cells).
Applying the 
homology functor $\Hm$ over all degrees
with coefficients in a field $\field$,
we obtain a {\it zigzag module}
\[
\Hm(\Fcal): 
\Hm(K_0) 
\leftrightarrow
\Hm(K_1) 
\leftrightarrow
\cdots 
\leftrightarrow
\Hm(K_\fcnt), \]
where
each $\Hm(K_i)\leftrightarrow \Hm(K_{i+1})$
is a linear map induced by inclusion.
In this paper, we present our ideas with $\field=\Zbb_2$
for simplicity;
some details regarding generalizing to arbitrary fields
are discussed
in \Cref{sec:gen-field}.
Based on~\cite{carlsson2010zigzag,Gabriel72},
the zigzag module $\Hm(\Fcal)$
has a decomposition of the form
\begin{equation}\label{eqn:inv-decomp}
\Hm(\Fcal)=\bigoplus_{k=1}^\ell\Ibb^{[\birth_k,\death_k]},
\end{equation}
in which each $\Ibb^{[\birth_k,\death_k]}$
is an \emph{interval submodule} of $\Hm(\Fcal)$.
The (multi-)set of intervals (bars)
$\Pers(\Fcal):=\Set{[\birth_k,\death_k]\given k=1,\ldots,\ell}$
is an invariant of $\Hm(\Fcal)$
and is called the \emph{zigzag barcode} of $\Fcal$.

Let $\Zyc(K_i)$ denote the cycle space (over all degrees) of $K_i$.
For a $[\birth_k,\death_k]\in\Pers(\Fcal)$,
a sequence of cycles 
\[z^k_{\birth_k}\in\Zyc(K_{\birth_k}),z^k_{\birth_k+1}\in\Zyc(K_{\birth_k+1}),\ldots,z^k_{\death_k}\in\Zyc(K_{\death_k})\]
is called its \emph{zigzag representative} 
if each $[z^k_i]$ generates $\Ibb^{[\birth_k,\death_k]}(i)$ for the interval summand
$\Ibb^{[\birth_k,\death_k]}$ in \Cref{eqn:inv-decomp}.
It follows:
(1) $[z^k_i]\mapsto[z^k_{i+1}]$ or $[z^k_i]\mapsfrom[z^k_{i+1}]$
based on the direction of the map;
(2) for each $K_i$, $\{[z^k_i]\mid i\in[\birth_k,\death_k]\}$
is a basis for $\Hm(K_i)$.

In other words, all zigzag representatives for bars in $\Pers(K_i)$
form \emph{compatible} bases for spaces in $\Hm(\filt)$
such that linear maps in $\Hm(\filt)$ are diagonalized.

An inclusion in a filtration is called \emph{cell-wise}
if it is an addition or deletion of a single cell $\sG$,
which we sometimes denote as, e.g., $K_i\leftrightarrowsp{\sG}K_{i+1}$.
A filtration is called \emph{cell-wise} if it contains only cell-wise inclusions.
For computation, 
zigzag filtrations in this paper are by default
cell-wise; we also assume they
start and end with empty complexes
(as adopted in~\cite{DBLP:conf/compgeom/Dey0M25,maria2014zigzag,maria2024discrete}).
We note that any filtration can be converted into this form
by expanding the inclusions and attaching complexes to both ends.
Zigzag representatives for general filtrations
can also be easily retrieved from those computed
on simplex-wise ones; see~\cite{DBLP:conf/compgeom/Dey0M25} for details.

\para{Matrix notations.}
For a matrix $A$, let $\matel{A}{i}{j}$ denote its element
on the $i$-th row and $j$-th column.
$\row{A}{i}$ denotes its $i$-th row;
$\col{A}{j}$, its $j$-th column.
The pivot of $\col{A}{j}$ is the row index of the lowest
non-zero entry of $\col{A}{j}$ and is denoted $\pivot(\col{A}{j})$.
Since columns and rows of matrices in this paper correspond to
$\DG$-cells, $i$ and $j$ could also be cells
and $\pivot(\col{A}{j})$ could also refer to the cell
corresponding to the row.

\subsection{Update operations\opt{siam}{.}}
\label{sec:update-oper}
We now present all the eight update operations
considered in the paper. Note that these operations are 
sufficient to transform any simplex-wise zigzag filtration into another~\cite{DBLP:conf/compgeom/DeyH24}.
In this subsection, $\Fcal$ and $\filtx$  are both 
simplex-wise
zigzag filtrations consisting of simplicial complexes
which start and end with $\emptyset$.

\medskip
\noindent
\textit{{Forward switch}}~\cite{maria2014zigzag}
swaps two forward inclusions and
requires $\sG\nsubseteq\tG$:
\vspace{-5pt}
\begin{equation}
\label{eqn:fwd-sw}
\begin{tikzpicture}[baseline=(current  bounding  box.center)]
\tikzstyle{every node}=[minimum width=24em]
\node (a) at (0,0) {$\Fcal:K_0 \leftrightarrow\cdots\leftrightarrow K_{i-1}\inctosp{\sG}K_i\inctosp{\tG}K_{i+1}\leftrightarrow\cdots\leftrightarrow K_\fcnt$}; 
\node (b) at (0,-0.6){$\filtx:K_0\leftrightarrow\cdots\leftrightarrow K_{i-1}\inctosp{\tG} \cplxx_i\inctosp{\sG} K_{i+1}\leftrightarrow\cdots\leftrightarrow K_\fcnt$};
\path[->] (a.0) edge [bend left=90,looseness=1.5,arrows={-latex},dashed] (b.0);
\end{tikzpicture}
\end{equation}

\medskip
\noindent
\textit{{Backward switch}}
swaps two backward inclusions and
requires $\tG\nsubseteq\sG$:
\vspace{-5pt}
\begin{equation}
\label{eqn:bak-sw}
\begin{tikzpicture}[baseline=(current  bounding  box.center)]
\tikzstyle{every node}=[minimum width=24em]
\node (a) at (0,0) {$\Fcal: K_0 \leftrightarrow
\cdots
\leftrightarrow 
K_{i-1}\bakinctosp{\sG} 
K_i 
\bakinctosp{\tG} K_{i+1}
\leftrightarrow
\cdots \leftrightarrow K_\fcnt$}; 
\node (b) at (0,-0.6){$\filtx: K_0 \leftrightarrow
\cdots
\leftrightarrow 
K_{i-1}\bakinctosp{\tG} 
\cplxx_i 
\bakinctosp{\sG} K_{i+1}
\leftrightarrow
\cdots \leftrightarrow K_\fcnt$};
\path[->] (a.0) edge [bend left=90,looseness=1.5,arrows={-latex},dashed] (b.0);
\end{tikzpicture}
\end{equation}

\noindent
\textit{{Outward/inward switch}}~\cite{carlsson2010zigzag}
swaps two inclusions of opposite directions
and requires $\sG\neq\tG$:
\vspace{-5pt}
\begin{equation}
\label{eqn:out-in-sw}
\begin{tikzpicture}[baseline=(current  bounding  box.center)]
\tikzstyle{every node}=[minimum width=24em]
\node (a) at (0,0) {$\Fcal: K_0 \leftrightarrow
\cdots
\leftrightarrow 
K_{i-1}\inctosp{\sG} 
K_i 
\bakinctosp{\tG} K_{i+1}
\leftrightarrow
\cdots \leftrightarrow K_\fcnt$}; 
\node (b) at (0,-0.6){$\filtx: K_0 \leftrightarrow
\cdots
\leftrightarrow 
K_{i-1}\bakinctosp{\tG} 
\cplxx_i 
\inctosp{\sG} K_{i+1}
\leftrightarrow
\cdots \leftrightarrow K_\fcnt$};
\path[->] (a.0) edge [bend left=90,looseness=1.5,arrows={latex-latex},dashed] (b.0);
\end{tikzpicture}
\end{equation}

\noindent
The switch from $\Fcal$ to $\filtx$ is 
an {{outward}} 
switch
and the switch from $\filtx$ to $\Fcal$ is  an {{inward}} switch.
Notice that if $\sG=\tG$, then 
e.g., for outward switch, 
we cannot delete $\tG$ from $K_{i-1}$ in $\filtx$
because $\tG\not\in K_{i-1}$.

\medskip
\noindent
\textit{{Inward contraction/expansion}}~\cite{maria2014zigzag}
is as follows:
\vspace{-5pt}
\begin{equation}
\label{eqn:in-contrac-expan}
\begin{tikzpicture}[baseline=(current  bounding  box.center)]
\node (a) at (0,0) {$\Fcal: K_0 \leftrightarrow
\cdots\leftrightarrow K_{i-2}
\leftrightarrow 
K_{i-1}\inctosp{\sG} 
K_i 
\bakinctosp{\sG} K_{i+1}
\leftrightarrow K_{i+2}\leftrightarrow
\cdots \leftrightarrow K_\fcnt$}; 
\node (b) at (0,-0.9){$\filtx: K_0 \leftrightarrow
\cdots\leftrightarrow K_{i-2}
\leftrightarrow 
\cplxx_{i}
\leftrightarrow K_{i+2}\leftrightarrow
\cdots \leftrightarrow K_\fcnt$};
\draw[->,arrows={latex-latex},dashed] (a.0) .. controls (+7,-0.1) and (+6.5,-0.8) .. (b.0);
\end{tikzpicture}
\end{equation}
\noindent
From $\Fcal$ to $\filtx$ we have an inward contraction 
and from $\filtx$ to $\Fcal$ we have an inward expansion.
We have $\cplxx_{i}=K_{i-1}=K_{i+1}$.
Because neither one of these complexes contains $\sG$, they also do not contain any of its cofaces.

\medskip
\noindent
\textit{{Outward contraction/expansion}}
is superficially similar to the inward version, but instead of removing successive addition and removal of a simplex, it fuses two different index-intervals over which the simplex appears in the filtration. The two center arrows now point away from each other:
\vspace{-5pt}
\begin{equation}
\label{eqn:out-contrac-expan}
\begin{tikzpicture}[baseline=(current  bounding  box.center)]
\node (a) at (0,0) {$\Fcal: K_0 \leftrightarrow
\cdots\leftrightarrow K_{i-2}
\leftrightarrow 
K_{i-1}\bakinctosp{\sG} 
K_i 
\inctosp{\sG} K_{i+1}
\leftrightarrow K_{i+2}\leftrightarrow
\cdots \leftrightarrow K_\fcnt$}; 
\node (b) at (0,-0.9){$\filtx: K_0 \leftrightarrow
\cdots\leftrightarrow K_{i-2}
\leftrightarrow 
\cplxx_{i}
\leftrightarrow K_{i+2}\leftrightarrow
\cdots \leftrightarrow K_\fcnt$};
\draw[->,arrows={latex-latex},dashed] (a.0) .. controls (+7,-0.1) and (+6.5,-0.8) .. (b.0);
\end{tikzpicture}
\end{equation}

\noindent Notice that we also have $\cplxx_{i}=K_{i-1}=K_{i+1}$.

\subsection{Conversion to non-zigzag filtration\opt{siam}{.}}
\label{sec:conversion}
We review results from prior works that lay the groundwork for this paper.
Specifically, we first describe
how a zigzag filtration $\filt$ can be converted to a non-zigzag filtration $\filtnz$~\cite{DBLP:conf/esa/DeyH22}
such that  there is a natural bijection between intervals in $\Pers(\filtnz)$ and those in $\Pers(\filt)$; see~\cite{DBLP:conf/esa/DeyH22}.
Equivalently, the interval mapping can be easily computed by tracking the pairing of simplices (cells) because a pair of  cells  determines an interval in $\Pers(\filtnz)$ if and only if the same pair (not necessarily in the same order) determines the corresponding 
interval in $\Pers(\filt)$.
We then present a recent finding~\cite{DBLP:conf/compgeom/Dey0M25} on how 
representatives for $\filtnz$ can be used to recover
representatives for the original $\filt$
from an $R=DV$ decomposition~\cite{cohen2006vines}
for $\filtnz$.
In \Cref{sec:ez-up}, we describe how to update the representatives after six of the eight operations described in \Cref{sec:update-oper}; the update rules follow directly from the conversion.

Given a \emph{simplex}-wise zigzag filtration 
\begin{equation*}
\filt:
\emptyset=
\cplx_0\leftrightarrowsp{\ssx_{0}} \cplx_1\leftrightarrowsp{\ssx_{1}}
\cdots 
\leftrightarrowsp{\ssx_{\fcnt-1}} \cplx_\fcnt
=\emptyset,
\end{equation*}
where a simplex can be inserted and deleted multiple
times,
we first treat each occurrence of a simplex in $\filt$ as a \emph{distinct}
$\DG$-cell, and have the following
\emph{cell}-wise zigzag filtration 
\begin{equation*}
\label{eqn:fzz-up-sec-filt}
\dfilt:
\emptyset=
\dcplx_0\leftrightarrowsp{\dssx_{0}} \dcplx_1\leftrightarrowsp{\dssx_{1}}
\cdots 
\leftrightarrowsp{\dssx_{\fcnt-1}} \dcplx_\fcnt
=\emptyset.
\end{equation*}
While $\dfilt$ ``looks'' exactly the same as 
$\filt$, treating the simplices as $\DG$-cells 
enables us to perform the subsequent conversion to a
\emph{non-zigzag} filtration
\begin{equation*}
\filtnz:
\{\oG\}=\ucplx_0\inctosp{\ussx_0}
\cdots\inctosp{\ussx_{\scnt{}-1}} 
\ucplx_{\scnt{}}
\inctosp{\cdec{\ussx}_{\scnt{}}}
{\ucplx}_{\scnt{}+1}
\cdots\inctosp{\cdec{\ussx}_{\fcnt-1}}
{\ucplx}_{\fcnt}.
\end{equation*}
Above, 
$\oG$ is an additional vertex for coning;
$\ussx_0,\ussx_1,\ldots,\ussx_{\scnt{}-1}$ are $\DG$-cells
inserted in $\dfilt$ with the same insertion order as in $\dfilt$;
$\ussx_{\scnt{}},\ussx_{\scnt{}+1},\ldots,\ussx_{\fcnt-1}$ are cones of $\DG$-cells
deleted in $\dfilt$ with the deletion order in $\dfilt$ reversed
to form the insertion order in $\filtnz$.
See \Cref{fig:convert} for an example of the conversion.
Notice that adjacency of cells in $\dfilt$
is carried over to $\filtnz$.

\begin{figure}[!tbh]
  \centering
  \opt{SoCG}{\includegraphics[width=\linewidth]{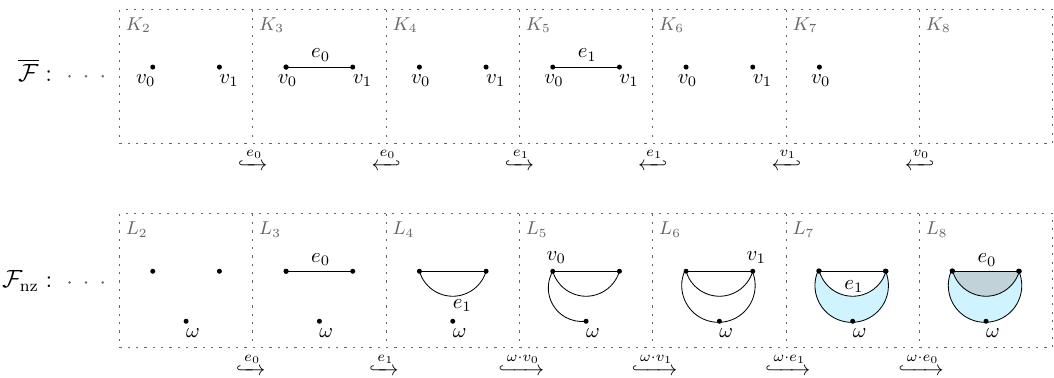}}
  \opt{arXiv,siam}{\includegraphics[width=0.95\linewidth]{fig/convert}}
  \caption{An example
  of converting a zigzag filtration of simplices to a non-zigzag filtration of $\DG$-cells,
  where $\filt$ and $\dfilt$ look the same and only $\dfilt$ is listed. 
  Note that $\filt,\filtnz$ both start with the complexes $\emptyset$ and $\{v_0\}$ which are omitted
    for brevity.
  The simplex $\{v_0,v_1\}$ is inserted twice
  in $\filt$ with the two insertions corresponding two different $\DG$-cells $e_0$, $e_1$ in $\dfilt$.
  Insertions of $e_0$, $e_1$ in $\dfilt$ correspond to the same insertions of them in the first half of $\filtnz$,
  while deletions of $e_0$, $e_1$, $v_1$, $v_0$ in $\dfilt$ corresponds to insertions of the cones
  $\oG\cdot v_0$, $\oG\cdot v_1$, $\oG\cdot e_1$, $\oG\cdot e_0$ in the second half of $\filtnz$, with the order reversed.}
  \label{fig:convert}
\end{figure}

\begin{figure}[!tbh]
  \centering
  \opt{SoCG}{\includegraphics[width=0.7\linewidth]{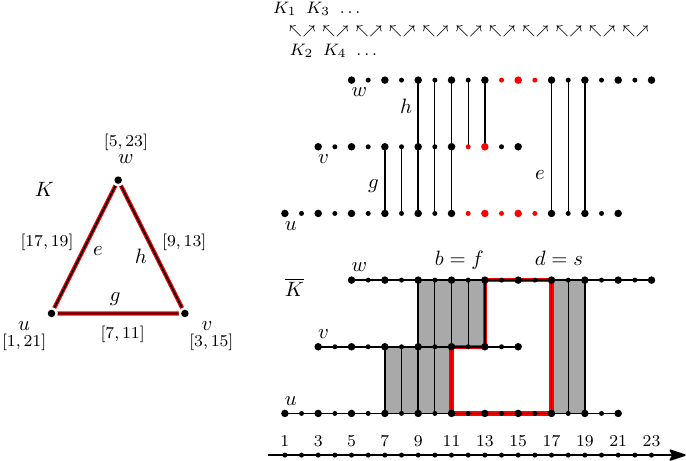}}
  \opt{arXiv,siam}{\includegraphics[width=0.65\linewidth]{fig/lifted-cycle}}
             
  \caption{An example of retrieving zigzag representatives from the converted non-zigzag filtration.
  The input zigzag filtration is shown at the top.
  The total (last) complex $K$ of the converted non-zigzag filtration is shown on the left, 
  with the labels indicating the lifetimes of simplices in the zigzag filtration.
  The constructed prism $\bar{K}$ and its levelset zigzag are shown at the bottom.
  The red cycle in $K$ is used to construct a representative for the zigzag bar $[12,16]$.
  The red cycle in $\bar{K}$ is the constructed apex representative.
  Slices of the apex representative provide the zigzag representative (vertices colored in red) for $[12,16]$.
  Figure from~\cite{DBLP:conf/compgeom/Dey0M25}.}
  \label{fig:lifted-cycle}
\end{figure}

\para{Retrieving barcode and representatives for $\filt$.}
Let $\Dmat$ 
be the filtered boundary matrix for $\filtnz$,
i.e., $\Dmat$ 
encodes the boundary operator $\partial$
with the cells ordered according to $\filtnz$.
Let $\Rmat$, $\Vmat$ be square matrices with the same dimension as $\Dmat$.
We call a decomposition of the form
\begin{equation}
\label{eqn:decomp-all}
\Rmat=\Dmat\Vmat
\end{equation}
\emph{valid} if $\Vmat$ is upper-triangular 
and has all 1's on the diagonal.
The decomposition is further called
\emph{reduced} if non-zero columns of $\Rmat$ have distinct pivots.
One can read off the barcode $\Pers(\filtnz)$
and the representatives from a reduced decomposition~\cite{de2011dualities}.
Moreover,
there is a bijection between $\Pers(\filt)$ and $\Pers(\filtnz)$
so that  $\Pers(\filt)$ could be recovered from $\Pers(\filtnz)$
with an overall complexity of $O(\fcnt)$ for the conversion~\cite{DBLP:conf/esa/DeyH22}.
Recently, the authors of~\cite{DBLP:conf/compgeom/Dey0M25} showed that one could
also recover the representatives for the  bars in $\Pers(\filt)$
from a reduced decomposition as in \Cref{eqn:decomp-all}.
They construct a \emph{prism} complex $\bar{K}$ with a real-valued function $\bar{f}$, whose
levelset zigzag~\cite{carlsson2009zigzag-realvalue}, denoted as $\LPers(\bar{f})$,
is isomorphic to the 
zigzag persistence of $\filt$, i.e., $\Pers(\filt)=\LPers(\bar{f})$.
They turn to the \emph{Mayer-Vetoris pyramid}~\cite{carlsson2009zigzag-realvalue} of function $\bar{f}$
and define \emph{apex representatives}, from which the representatives of $\Pers(\filt)$ can be efficiently retrieved
(see \Cref{fig:lifted-cycle} for an example).
Their main result is summarized in the following theorem.
\begin{theorem}[\hspace{-0.01pt}\cite{DBLP:conf/compgeom/Dey0M25}]
\label{thm:extr-zzrep}
Given a reduced decomposition $R = DV$ for $\filtnz$,
for each $\dimx$-dimensional bar $[b,d]\in\Pers(\filt)$ and its corresponding persistence pair
$(\ssx,\tsx)$ in $\Pers(\filtnz)$,
one can retrieve the apex representative $\GG$ for the corresponding bar
in $\LPers(\bar{f})$ from $R[\tsx]$ or $V[\tsx]$
in $O(\dimx\cdot\fcnt\log\fcnt)$ time.
Moreover, the zigzag representative for $[b,d]\in\Pers(\filt)$ can be retrieved
from $\GG$ in $O(\log\fcnt+C)$ time, where $C$ is the size of the output.
\end{theorem}

According to \Cref{thm:extr-zzrep}, 
each interval in $\Pers(\filt)$ corresponds
to a column in $R$ or $V$. The correspondence can be easily tracked 
over the update operations.
As mentioned, 
we perform update on the corresponding $R=DV$ decomposition.
When the corresponding column in  $R$ or $V$ has changed during the update,
one could always extract the zigzag representative
for an interval in $O(\log\fcnt+C)$ time
when necessary.

\section{Easy updates\opt{siam}{.}}
\label{sec:ez-up}
In this section
we present the update for six of the operations
in \Cref{sec:update-oper}.
The update algorithms for outward contraction and expansion
involving the adjacency change are presented in \Cref{sec:out-contra,sec:out-exp} respectively.
For all the updates, we assume that we are given a reduced decomposition as in \Cref{eqn:decomp-all}
for the non-zigzag filtration 
that corresponds to the zigzag filtration 
before the operation,
and the goal is to obtain a reduced decomposition
for the non-zigzag filtration 
after the operation.
Representatives for the zigzag filtrations
can then be extracted 
using \Cref{thm:extr-zzrep}.

\para{Outward/inward switch.}
Since the switch in \Cref{eqn:out-in-sw}
swaps two inclusions of different directions,
the corresponding non-zigzag filtrations before and after
the switch are the same.
So no update to the $R=DV$ decomposition is necessary.

\para{Forward/backward switch.}
Corresponding to a forward switch in \Cref{eqn:fwd-sw} on the original zigzag filtration,
there is a forward switch 
in the first half of the non-zigzag filtration consisting of the base cells.
Similarly,
corresponding to a backward switch in \Cref{eqn:bak-sw} on the zigzag filtration,
there is a forward switch 
in the second half of the non-zigzag filtration consisting of the cone cells.
To perform the update for a forward switch
on a non-zigzag filtration,
we utilize the $O(\fcnt)$ {transposition algorithm}
in~\cite{cohen2006vines}.

\para{Inward contraction/expansion.}
Let $\filtnz$ (resp.\ $\filtxnz$) be
the non-zigzag filtration corresponding to $\filt$ (resp.\ $\filtx$)
in \Cref{eqn:in-contrac-expan}.
Also,
let $\ssx_0$ be the $\DG$-cell in $\filtnz$
corresponding to the insertion $K_{i-1}\inctosp{\ssx}K_{i}$ in $\filt$
and let $\cdec{\ssx}_0$ be the cone of $\ssx_0$  corresponding to the deletion $K_{i}\bakinctosp{\ssx}K_{i+1}$
in $\filt$.
Since $\ssx$ is deleted immediately after it is inserted to $K_{i}$ in $\filt$,
we have that $\ssx_0$ has no cofaces in $\filtnz$.
This also implies that $\cdec{\ssx}_0$ has no cofaces in $\filtnz$.
To perform the update for an inward expansion, suppose that we
are given a reduced decomposition $R^{*}=D^{*}V^{*}$ for $\filtxnz$ before the expansion.
We first append columns representing the chains $\ssx_0$, $\cdec{\ssx}_0$ to $V^{*}$,
and columns representing $\partial\ssx_0$, $\partial\cdec{\ssx}_0$ to $R^{*}$.
Notice that we also need to append new rows corresponding to $\ssx_0$, $\cdec{\ssx}_0$
to $R^{*}$ and $V^{*}$.
Call the newly derived matrices $R$ and $V$.
We then perform two column reductions~\cite{edelsbrunner2000topological} 
to make $R$ reduced,
which takes $O(\fcnt^2)$ time.
After this,
we perform $O(\fcnt)$ row/column transpositions on $R$ and $V$ using the algorithm in~\cite{cohen2006vines}
to place the columns and rows of $\ssx_0$, $\cdec{\ssx}_0$
to the appropriate positions 
as in $\filtnz$. 
Thus, the update for inward expansion takes $O(\fcnt^2)$ time.

We now describe the update for an inward contraction.
Suppose that we
are given a reduced decomposition $R=DV$ for $\filtnz$.
We first perform $O(\fcnt)$ transpositions~\cite{cohen2006vines} so that the columns (resp.\ rows) corresponding to 
$\ssx_0$ and $\cdec{\ssx}_0$ are the last two columns (resp.\ rows) in $R$ and $V$.
Notice that we could do this because 
$\ssx_0$, $\cdec{\ssx}_0$  have no cofaces in $\filtnz$.
We then simply remove
the columns and rows of
$\ssx_0$ and $\cdec{\ssx}_0$ from $R$ and $V$.
Accordingly, the update for the inward contraction also takes $O(\fcnt^2)$ time.
\begin{remark}
The fact that $\ssx_0$ and $\cdec{\ssx}_0$ have no cofaces in the converted non-zigzag
filtration is a key condition 
that allows existing algorithms to be applied.
The outward contraction and expansion, however, do not have such a property, which makes them more difficult.
In fact, there is certain non-trivial change on the cell adjacency
over the two operations, which is 
detailed in \Cref{sec:bound-chg}.
\end{remark}

\section{Outward contraction\opt{siam}{.}}\label{sec:out-contra}

We present the update algorithm for an outward contraction.
Recall that an outward contraction is the following operation:
\begin{equation}
\label{eqn:out-contrac-in-sec}
\begin{tikzpicture}[baseline=(current  bounding  box.center)]
\node (a) at (0,0) {$\filt: \cplx_0 \leftrightarrow
\cdots\leftrightarrow \cplx_{i-2}
\leftrightarrow 
\cplx_{i-1}\bakinctosp{\sG} 
\cplx_i 
\inctosp{\sG} \cplx_{i+1}
\leftrightarrow \cplx_{i+2}\leftrightarrow
\cdots \leftrightarrow \cplx_\fcnt$}; 
\node (b) at 
(0,-0.95){$\filtx: \cplx_0 \leftrightarrow
\cdots\leftrightarrow \cplx_{i-2}
\leftrightarrow 
\cplxx_{i}
\leftrightarrow \cplx_{i+2}\leftrightarrow
\cdots \leftrightarrow \cplx_\fcnt$};
\path[->] ([xshift=7pt]a.south) edge[double distance=1.5pt,arrows=-{Classical TikZ Rightarrow[length=3pt]}]  ([xshift=7pt]b.north);
\end{tikzpicture}
\end{equation}
where $\cplxx_{i}=\cplx_{i-1}=\cplx_{i+1}$
and $\ssx$ is a $p$-simplex.

Let $\filtnz$ and $\filtxnz$, respectively, be
the non-zigzag filtrations corresponding to $\filt$ and $\filtx$
by the conversion in \Cref{sec:conversion}.
Rather than considering a decomposition over all degrees as in \Cref{eqn:decomp-all}, 
we work on a separate decomposition for each degree.
Specifically,
for each $\dimx$,
let $\Dmat_\dimx$ (resp.\ $\Dmatx_\dimx$) 
be the $\dimx$-th filtered boundary matrix for $\filtnz$ (resp.\ $\filtxnz$)
which encodes the boundary operator $\partial_\dimx:\Chn_\dimx\to\Chn_{\dimx-1}$.
We assume that we are given a reduced decomposition
\begin{equation*}
\Rmat_\dimx=\Dmat_\dimx\Vmat_\dimx
\end{equation*}
for each $\dimx$,
where $\Rmat_\dimx$ has the same dimension as $\Dmat_\dimx$
(i.e., rows corresponding
to $(\dimx-1)$-cells and columns corresponding to $\dimx$-cells),
and $\Vmat_\dimx$ is a square matrix with both rows and columns corresponding to
$\dimx$-cells.
The goal of the update is to compute the decompositions for $\filtxnz$.

\subsection{Changes on the boundary matrices\opt{siam}{.}}\label{sec:bound-chg}
To present our update strategy,
we first describe the changes from 
the boundary matrix $\Dmat_{\dimx}$ to $\Dmatx_{\dimx}$
for the different $\dimx$.
In a zigzag filtration, each insertion of a simplex $\xssx$ is associated with a deletion
of $\xssx$, which 
determines the time interval of a specific occurrence of $\xssx$.
Also notice that a coface of $\xssx$ must have its time intervals contained
in those of $\xssx$.
In \Cref{eqn:filt-d}, we illustrate the corresponding insertion and deletion 
for the two arrows $\cplx_{i-1}\bakinctosp{\ssx}\cplx_i\inctosp{\ssx}\cplx_{i+1}$:
\begin{equation}
\label{eqn:filt-d}
\opt{arXiv,siam}{\scalebox{1.1}}{\begin{tikzpicture}[baseline=(current  bounding  box.center)]
{\opt{arXiv,siam}{\small }\node (a) at (0,0) {$\filt: 
\cdots
\inctosp{\ssx}
\cdots
\inctosp{\tsx} 
\cdots
\bakinctosp{\tsx}
\cdots
\leftrightarrow 
\cplx_{i-1}\bakinctosp{\ssx} 
\cplx_i 
\inctosp{\ssx} \cplx_{i+1}
\leftrightarrow 
\cdots
\inctosp{\tsx'} 
\cdots
\bakinctosp{\tsx'}
\cdots
\bakinctosp{\ssx}
\cdots$}; 
\node (b) at (0,-1.5){$\dfilt:\cdots
\inctosp{\ssx_1}
\cdots
\inctosp{\tsx_0} 
\cdots
\bakinctosp{\tsx_0}
\cdots
\leftrightarrow  
\dcplx_{i-1}\bakinctosp{\ssx_1} 
\dcplx_i
\inctosp{\ssx_2}\dcplx_{i+1}
\leftrightarrow 
\cdots
\inctosp{\tsx'_0} 
\cdots
\bakinctosp{\tsx'_0}
\cdots
\bakinctosp{\ssx_2}
\cdots$};}
\path[->] ([xshift=7pt,yshift=-1pt]a.south) edge[double distance=1.5pt,arrows=-{Classical TikZ Rightarrow[length=3pt]}]  ([xshift=7pt,yshift=9pt]b.north);
\draw (-5.17,0.55) -- (-0.4,0.55); %
\draw (-5.17,0.55) -- (-5.17,0.25); %
\draw (-0.4,0.55) -- (-0.4,0.25); %
\draw (-4.05,0.4) -- (-2.85,0.4); %
\draw (-4.05,0.4) -- (-4.05,0.25); %
\draw (-2.85,0.4) -- (-2.85,0.25); %
\draw (5.75,0.55) -- (0.75,0.55); %
\draw (0.75,0.55) -- (0.75,0.25); %
\draw (5.75,0.55) -- (5.75,0.25); %
\draw (4.5,0.4) -- (3.25,0.4); %
\draw (3.25,0.4) -- (3.25,0.25); %
\draw (4.5,0.4) -- (4.5,0.25); %
\draw (-5.6,-0.95) -- (-0.4,-0.95); %
\draw (-5.6,-0.95) -- (-5.6,-1.25); %
\draw (-0.4,-0.95) -- (-0.4,-1.25); %
\draw (-4.3,-1.1) -- (-3,-1.1); %
\draw (-4.3,-1.1) -- (-4.3,-1.25); %
\draw (-3,-1.1) -- (-3,-1.25); %
\draw (0.9,-0.95) -- (6.05,-0.95); %
\draw (0.9,-0.95) -- (0.9,-1.25); %
\draw (6.05,-0.95) -- (6.05,-1.25); %
\draw (3.5,-1.1) -- (4.75,-1.1); %
\draw (3.5,-1.1) -- (3.5,-1.25); %
\draw (4.75,-1.1) -- (4.75,-1.25); %
\end{tikzpicture}}
\end{equation}
We suppose there are $(\Dim+1)$-cofaces $\tsx$, $\tsx'$ of $\ssx$ 
whose time intervals are contained in the two intervals of $\ssx$
respectively.
In \Cref{eqn:filt-d},
we also illustrate the filtration $\dfilt$ consisting of the $\DG$-complexes as in \Cref{sec:conversion},
where $\ssx_1$, $\ssx_2$ are $\DG$-cells
corresponding to
the two occurrences of $\ssx$ in $\filt$ respectively,
and  $\tsx_0$, $\tsx'_0$ are $\DG$-cells corresponding to
$\tsx$, $\tsx'$ respectively.
By the conversion in \Cref{sec:conversion}, $\ssx_1$ is a $\Dim$-face of $\tsx_0$
and $\ssx_2$ is a $\Dim$-face of $\tsx'_0$
in $\dfilt$.
Notice that the adjacency is also carried over to the non-zigzag filtration $\filtnz$.

\Cref{eqn:filtx-d} illustrates the adjacency of the  simplices/cells
in $\filtx$ and $\dfiltx$ after the contraction:
\begin{equation}
\label{eqn:filtx-d}
\opt{arXiv,siam}{\scalebox{1.1}}{\begin{tikzpicture}[baseline=(current  bounding  box.center)]
{\opt{arXiv,siam}{\small }\node (a) at (0,0) {$\filtx: 
\cdots
\inctosp{\ssx}
\cdots
\inctosp{\tsx} 
\cdots
\bakinctosp{\tsx}
\cdots
\leftrightarrow 
\cplxx_i 
\leftrightarrow 
\cdots
\inctosp{\tsx'} 
\cdots
\bakinctosp{\tsx'}
\cdots
\bakinctosp{\ssx}
\cdots$}; 
\node (b) at (0,-1.5){$\dfiltx:\cdots
\inctosp{\ssx_0}
\cdots
\inctosp{\tsx_0} 
\cdots
\bakinctosp{\tsx_0}
\cdots
\leftrightarrow  
\dcplxx
\leftrightarrow 
\cdots
\inctosp{\tsx'_0} 
\cdots
\bakinctosp{\tsx'_0}
\cdots
\bakinctosp{\ssx_0}
\cdots$};}
\path[->] ([xshift=7pt,yshift=-1pt]a.south) edge[double distance=1.5pt,arrows=-{Classical TikZ Rightarrow[length=3pt]}]  ([xshift=7pt,yshift=9pt]b.north);
\draw (-3.6,0.55) -- (4.35,0.55); %
\draw (-3.6,0.55) -- (-3.6,0.25); %
\draw (4.35,0.55) -- (4.35,0.25); %
\draw (-2.45,0.4) -- (-1.3,0.4); %
\draw (-2.45,0.4) -- (-2.45,0.25); %
\draw (-1.3,0.4) -- (-1.3,0.25); %
\draw (3.1,0.4) -- (1.85,0.4); %
\draw (1.85,0.4) -- (1.85,0.25); %
\draw (3.1,0.4) -- (3.1,0.25); %
\draw (-3.85,-0.95) -- (4.45,-0.95); %
\draw (-3.85,-0.95) -- (-3.85,-1.25); %
\draw (4.45,-0.95) -- (4.45,-1.25); %
\draw (-2.55,-1.1) -- (-1.3,-1.1); %
\draw (-2.55,-1.1) -- (-2.55,-1.25); %
\draw (-1.3,-1.1) -- (-1.3,-1.25); %
\draw (1.95,-1.1) -- (3.2,-1.1); %
\draw (1.95,-1.1) -- (1.95,-1.25); %
\draw (3.2,-1.1) -- (3.2,-1.25); %
\end{tikzpicture}}
\end{equation}
The insertion of $\ssx$ before $\cplx_{i-1}\bakinctosp{\ssx}\cplx_i\inctosp{\ssx}\cplx_{i+1}$
in $\filt$ and the deletion of $\ssx$ after, now become the endpoints of the same time interval in $\filtx$.
As a result, the $\DG$-cell $\ssx_0$ corresponding to $\ssx$ now becomes a $\Dim$-face
of $\tsx_0$, $\tsx'_0$ in $\dfiltx$ as well as in $\filtxnz$.
The above change
applies to any $(\Dim+1)$-cofaces of $\ssx_1$, $\ssx_2$ in $\dfilt$ and $\filtnz$.
In other words, the two $\Dim$-cells $\ssx_1$, $\ssx_2$ in $\filtnz$ are identified
as the same $\Dim$-cell $\ssx_0$ in $\filtxnz$.
Since cell insertions in the first half of $\filtxnz$ follow the same order as
cell insertions in $\dfiltx$,
$\ssx_0$ in $\filtxnz$ corresponds
to $\ssx_1$ in $\filtnz$.
The change in the second halves of
$\filtnz$ and $\filtxnz$
which involve the cone cells is similar. However, since the insertion order
of cone cells in  $\filtnz$ (resp.\ $\filtxnz$) is the reverse of the deletion order of
cells in $\dfilt$ (resp.\ $\dfiltx$), 
the cone $\cdec{\ssx}_0$ of ${\ssx}_0$ in $\filtxnz$
now corresponds
to the cone $\cdec{\ssx}_2$ of ${\ssx}_2$ in $\filtnz$.

\begin{figure}[!tb]
  \centering
  \opt{SoCG}{\includegraphics[width=\linewidth]{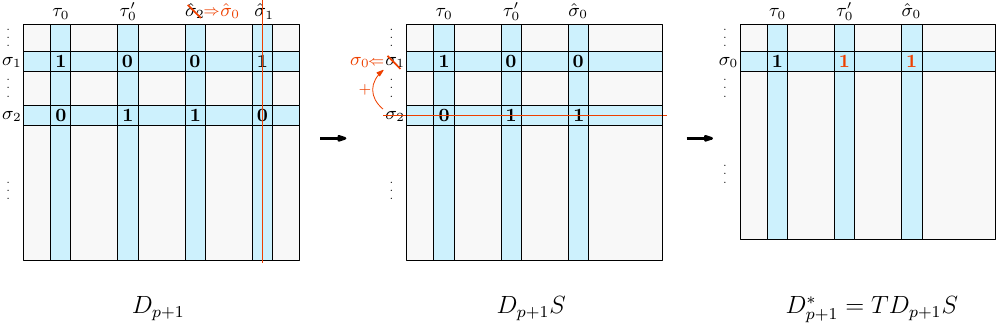}}
  \opt{arXiv,siam}{\includegraphics[width=0.9\linewidth]{fig/mat_adj_chg}}
  \caption{Change from boundary matrix $\Dmat_{\Dim+1}$ to $\Dmatx_{\Dim+1}$ (see \Cref{eqn:dmat-mult}
  for definitions of $\Tmat$ and $\Smat$).}
  \label{fig:mat_adj_chg}
\end{figure}

We summarize the change as follows:
\begin{observation}\label{obsv:contra-adj-change}
From $\filtnz$ to $\filtxnz$,
the $\Dim$-cell $\ssx_2$ and the $(\Dim+1)$-cell $\cdec{\ssx}_1$
are removed,
the $\Dim$-cell $\ssx_1$ and the $(\Dim+1)$-cell $\cdec{\ssx}_2$ are renamed as $\ssx_0$ and $\cdec{\ssx}_0$ respectively,
and $(\Dim+1)$-cofaces of $\ssx_1$, $\ssx_2$
{\rm(}resp.\ $(\Dim+2)$-cofaces of $\cdec{\ssx}_1$, $\cdec{\ssx}_2${\rm)}
now have $\ssx_0$ {\rm(}resp.\ $\cdec{\ssx}_0${\rm)} as face.
Accordingly, the change from the boundary matrix
$\Dmat_{\dimx}$ to $\Dmatx_{\dimx}$ for
the different $\dimx$
can be described as follows: 

\begin{itemize}
    \item $\Dmat_{\Dim+1}$ to $\Dmatx_{\Dim+1}${\rm:} 
    the column of $\cdec{\ssx}_1$ is removed,
and column label $\cdec{\ssx}_2$ is changed to $\cdec{\ssx}_0$;
    the row of $\ssx_2$ is added to
the row of $\ssx_1$, with 
the row of $\ssx_2$ removed, and
the row label $\ssx_1$ changed to $\ssx_0$;
see \Cref{fig:mat_adj_chg}.
    \item $\Dmat_{\Dim}$ to $\Dmatx_{\Dim}${\rm:} 
    the column of $\ssx_2$ is removed, and the column label $\ssx_1$ is changed to $\ssx_0$.
    \item $\Dmat_{\Dim+2}$ to $\Dmatx_{\Dim+2}${\rm:} 
    the row of $\cdec{\ssx}_1$ is added to
the row of $\cdec{\ssx}_2$, with
the row of $\cdec{\ssx}_1$ removed,
and the row label $\cdec{\ssx}_2$ changed to $\cdec{\ssx}_0$.
\end{itemize}

\noindent
We also have $\Dmat_{\dimx}=\Dmatx_{\dimx}$ for any other $\dimx$.
\end{observation}

\subsection{Update for degree $\Dim+1$\opt{siam}{.}}\label{sec:contra-up}

We describe in detail the update for degree $\Dim+1$, 
i.e., given the decomposition 
\[\Rmat_{\Dim+1}=\Dmat_{\Dim+1}\Vmat_{\Dim+1}\]
for $\filtnz$,
we find the corresponding decomposition for $\filtxnz$.
The updates for the degrees $\Dim$, $\Dim+2$ are similar but simpler; they are briefly described in \Cref{sec:contra-up-other}.

\begin{figure}[!tb]
  \centering
  \opt{SoCG}{\includegraphics[width=.7\linewidth]{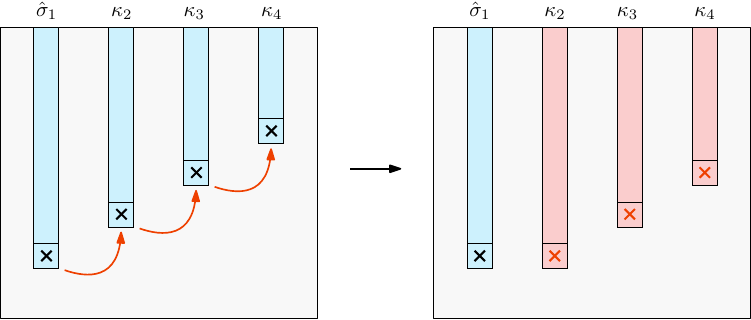}}
  \opt{arXiv,siam}{\includegraphics[width=.6\linewidth]{fig/pareto}}
  \caption{An illustration of columns in $\Rmat_{\Dim+1}$ 
  whose corresponding columns in 
  $\Vmat_{\Dim+1}$ have non-zero values on $\cdec{\ssx}_1$
  in Step I.2, where $\ell=4$.
  Pivots of these columns in $\Rmat_{\Dim+1}$ are monotonically decreasing before the column summation
  in Step I.2, with the red arrows indicating the direction of the summation.
  Note that we perform the summation \emph{starting from the right-most two columns}
  so that after this, only pivot of the column of $\cdec{\ssx}_1$
  may clash with another one. The clash is fixed later on.}
  \label{fig:pareto}
\end{figure}

\para{Step I.}
We first perform column summations on $\Vmat_{\Dim+1}$, and correspondingly on
$\Rmat_{\Dim+1}$, to make the row $\row{\Vmat_{\Dim+1}}{\cdec{\ssx}_1}$
contain zero entries everywhere except on the diagonal.
The aim is to remove the 
reference to $\cdec{\ssx}_1$ in $\Vmat_{\Dim+1}$ 
to prepare
for the later removal of $\cdec{\ssx}_1$.
To ensure a valid decomposition, we only perform \emph{left-to-right} column
summations so that $\Vmat_{\Dim+1}$ is still an upper triangular matrix with 1's on the diagonal.
We also ensure $\Rmat_{\Dim+1}$ has distinct column pivots almost everywhere,
possibly with a single violation which will be fixed later.

There are two cases for the step. First, if $\col{\Rmat_{\Dim+1}}{\cdec{\ssx}_1}=0$, i.e.,
$\col{\Vmat_{\Dim+1}}{\cdec{\ssx}_1}$ is a cycle, then we only need to
add $\col{\Vmat_{\Dim+1}}{\cdec{\ssx}_1}$ to any other $\col{\Vmat_{\Dim+1}}{\xssx}$
with $\matel{\Vmat_{\Dim+1}}{\cdec{\ssx}_1}{\xssx}\neq 0$,
without needing to change $\col{\Rmat_{\Dim+1}}{\xssx}$.

Now suppose that $\col{\Rmat_{\Dim+1}}{\cdec{\ssx}_1}\neq 0$.
We have two separate sub-steps
(note that these two sub-steps bear some similarity to the SiRUP algorithm~\cite{giunti2023pruning}
that considers removing a simplex without cofaces from a non-zigzag filtration).

\begin{itemize}

    \item Step I.1:
        While
    there are two columns $\col{\Vmat_{\Dim+1}}{\xssx}$ and $\col{\Vmat_{\Dim+1}}{\xssxx}$,
    $\xssx<\xssxx$, with
    non-zero entries on $\cdec{\ssx}_1$ such that 
    \[\pivot(\col{\Rmat_{\Dim+1}}{\xssx})<\pivot(\col{\Rmat_{\Dim+1}}{\xssxx}),\]
    add $\col{\Vmat_{\Dim+1}}{\xssx}$ to $\col{\Vmat_{\Dim+1}}{\xssxx}$
    and $\col{\Rmat_{\Dim+1}}{\xssx}$ to $\col{\Rmat_{\Dim+1}}{\xssxx}$,
    so that $\matel{\Vmat_{\Dim+1}}{\cdec{\ssx}_1}{\xssxx}=0$
    and the pivot of $\col{\Rmat_{\Dim+1}}{\xssxx}$ stays the same.
    
    \item Step I.2:
    Let
    \[\col{\Vmat_{\Dim+1}}{\xssx_1},\ldots,\col{\Vmat_{\Dim+1}}{\xssx_\ell}\]
    be all columns in $\Vmat_{\Dim+1}$
    that still have non-zero entries on $\cdec{\ssx}_1$,
    where \[\xssx_1=\cdec{\ssx}_1<\xssx_2<\cdots<\xssx_\ell.\]
    We have \[{\pivot}(\col{\Rmat_{\Dim+1}}{\xssx_1})>\cdots>\pivot(\col{\Rmat_{\Dim+1}}{\xssx_\ell})\]
    because if $\pivot(\col{\Rmat_{\Dim+1}}{\xssx_j})<\pivot(\col{\Rmat_{\Dim+1}}{\xssx_{j+1}})$
    for a $j$,
    then $\matel{\Vmat_{\Dim+1}}{\cdec{\ssx}_1}{\xssx_{j+1}}$ would have been set 
    to zero in Step I.1.
    Also notice that $\xssx_1,\ldots,\xssx_\ell$ are at the \emph{Pareto front}
    of the partial order defined by the order 
    of column indices and the order of pivots in $\Rmat_{\Dim+1}$. 
    See \Cref{fig:pareto} for an example of $\ell=4$.
    
    We then do the following:
    \begin{itemize}
        \item For $j=\ell,\ell-1,\ldots,2$,
        add $\col{\Rmat_{\Dim+1}}{\xssx_{j-1}}$
        to $\col{\Rmat_{\Dim+1}}{\xssx_j}$ and
        $\col{\Vmat_{\Dim+1}}{\xssx_{j-1}}$ to
        $\col{\Vmat_{\Dim+1}}{\xssx_j}$.
    \end{itemize}
\end{itemize}

\begin{observation}\label{obsv:I-2-pivots}
After Step I.2,
the row
$\row{\Vmat_{\Dim+1}}{\cdec{\ssx}_1}$ has zero entries everywhere
except on the diagonal.
Moreover, for each $j$, the pivot of $\col{\Rmat_{\Dim+1}}{\xssx_j}$ 
after Step I.2 equals the pivot of
$\col{\Rmat_{\Dim+1}}{\xssx_{j-1}}$ before Step I.2.
Hence, columns of $\Rmat_{\Dim+1}$ have distinct pivots everywhere
except $\col{\Rmat_{\Dim+1}}{\cdec{\ssx}_{1}}$,
for which $\pivot(\col{\Rmat_{\Dim+1}}{\cdec{\ssx}_{1}})=\pivot(\col{\Rmat_{\Dim+1}}{\xssx_{2}})$.
\end{observation}
\begin{remark}
If $\ell=1$ in Step I.2, 
then no column summations are needed,
and columns of $\Rmat_{\Dim+1}$ have distinct pivots
after Step I.2.
\end{remark}

\para{Step II.}
Notice that the change from $\Dmat_{\Dim+1}$ to $\Dmatx_{\Dim+1}$
involves row and column operations that can be expressed as multiplication of 
$\Dmat_{\Dim+1}$ on the left and right by suitable matrices.
Specifically,
let $\scnt{\dimx}$, $\scntx{\dimx}$ be the number of $\dimx$-cells
in $\filtnz$, $\filtxnz$, respectively. We have
\begin{equation}\label{eqn:dmat-mult}
\Dmatx_{\Dim+1}=\Smat\Dmat_{\Dim+1}\Tmat,
\end{equation}
where $\Smat$ is an $\scntx{\Dim}\times\scnt{\Dim}$ matrix representing
the following row operations: 

\begin{itemize}
    \item 
    $\row{\Dmatx_{\Dim+1}}{\ssx_0}=\row{\Dmat_{\Dim+1}}{\ssx_1}+\row{\Dmat_{\Dim+1}}{\ssx_2}$;
    \item
    $\row{\Dmatx_{\Dim+1}}{\xssx}=\row{\Dmat_{\Dim+1}}{\xssx}$ for any $\xssx\neq\ssx_0$ in $\filtxnz$.
\end{itemize}
Similarly, $\Tmat$ is an $\scnt{\Dim+1}\times\scntx{\Dim+1}$ matrix
which removes the column of $\cdec{\ssx}_1$ from $\Dmat_{\Dim+1}$.
See \Cref{fig:mat_adj_chg}.

Now let 
\begin{equation}\label{eqn:rvmat-mult}
\Rmatx_{\Dim+1}=\Smat\Rmat_{\Dim+1}\Tmat\text{ and }\Vmatx_{\Dim+1}=\Tmat^\tran\Vmat_{\Dim+1}\Tmat,
\end{equation}
where $\Tmat^\tran$ removes the row of $\cdec{\ssx}_1$ from $\Vmat_{\Dim+1}$.
We have:
\begin{proposition}\label{prop:Rx-DxVx}
$\Rmatx_{\Dim+1}=\Dmatx_{\Dim+1}\Vmatx_{\Dim+1}$.
\end{proposition}
\begin{proof}
By \Cref{eqn:dmat-mult,eqn:rvmat-mult},
we need to show     
\[\Smat\Rmat_{\Dim+1}\Tmat=\Smat\Dmat_{\Dim+1}\Tmat\Tmat^\tran\Vmat_{\Dim+1}\Tmat,\]
which is to show
\begin{equation}\label{eqn:RT-DTTVT}
\Rmat_{\Dim+1}\Tmat=\Dmat_{\Dim+1}\Tmat\Tmat^\tran\Vmat_{\Dim+1}\Tmat.
\end{equation}
Let $I^\circ=\Tmat\Tmat^\tran$. We have that the only difference between $I^\circ$
and an identity matrix is that the diagonal entry in $I^\circ$ corresponding to $\cdec{\ssx}_1$
is zero; see \Cref{fig:RT-DTTVT}.
Multiplying by $I^\circ$ on the left of $\Vmat_{\Dim+1}$ 
then sets the row of $\cdec{\ssx}_1$ to zero.
However, since only the diagonal entry in the row 
$\row{\Vmat_{\Dim+1}}{\cdec{\ssx}_1}$ is non-zero,
$\Vmat_{\Dim+1}$ and $I^\circ\Vmat_{\Dim+1}$ differ only in 
 the columns corresponding to $\cdec{\ssx}_1$.
Since $\Rmat_{\Dim+1}=\Dmat_{\Dim+1}\Vmat_{\Dim+1}$,
we also have that $\Rmat_{\Dim+1}$ and $\Dmat_{\Dim+1}I^\circ\Vmat_{\Dim+1}$ 
differ only in the columns corresponding to $\cdec{\ssx}_1$.
Since multiplying by $\Tmat$ on the right removes the column 
corresponding to $\cdec{\ssx}_1$, we then have \Cref{eqn:RT-DTTVT}.
\end{proof}

\Cref{prop:Rx-DxVx} provides the initial decomposition for $\filtxnz$ that we work on.
Notice that while $\Vmatx_{\Dim+1}$ is still upper triangular with all 1's on the diagonal,
$\Rmatx_{\Dim+1}$ may not necessarily have distinct column pivots.
By \Cref{obsv:I-2-pivots} and the fact that the only difference 
between $\Rmat_{\Dim+1}\Tmat$ and $\Rmat_{\Dim+1}$ is the removal of the column 
corresponding to $\cdec{\ssx}_1$, $\Rmat_{\Dim+1}\Tmat$ has distinct
column pivots. However, multiplying $\Smat$ with $\Rmat_{\Dim+1}\Tmat$, that is, adding the row of $\ssx_2$ to the row of $\ssx_1$ and 
removing the row of $\ssx_2$, may change the pivot of a certain column. Specifically, we have that 
the pivot of a column changes from $\Rmat_{\Dim+1}\Tmat$ to $\Rmatx_{\Dim+1}$
if and only if the pivot of the column is $\ssx_2$. 
To see this, we notice the following: (1)~For a column in $\Rmat_{\Dim+1}\Tmat$ whose pivot is greater than $\ssx_2$,
pivot of the corresponding column in $\Rmatx_{\Dim+1}$ stays the same;
(2)~For a column in $\Rmat_{\Dim+1}\Tmat$ whose pivot is less than $\ssx_2$,
since the entry of $\ssx_2$ in the column is 0,
the corresponding column in $\Rmatx_{\Dim+1}$ stays the same.

Then, there are at most two columns in $\Rmatx_{\Dim+1}$ 
having the same pivot, and we apply the following \emph{cascade sum} 
 on $\Rmatx_{\Dim+1}$ and $\Vmatx_{\Dim+1}$ to make the column pivots in
$\Rmatx_{\Dim+1}$ distinct:

\begin{itemize}
    \item
    While
    there are two columns $\col{\Rmatx_{\Dim+1}}{\xssx}$ and $\col{\Rmatx_{\Dim+1}}{\xssxx}$,
    $\xssx<\xssxx$, with the same pivot,
    add $\col{\Rmatx_{\Dim+1}}{\xssx}$ to $\col{\Rmatx_{\Dim+1}}{\xssxx}$
    and $\col{\Vmatx_{\Dim+1}}{\xssx}$ to $\col{\Vmatx_{\Dim+1}}{\xssxx}$.
\end{itemize}

Notice that each time we change the pivot of a column of $\Rmatx_{\Dim+1}$ in the above loop, its 
pivot row index becomes smaller, and again \emph{no more than two columns in $\Rmatx_{\Dim+1}$
have the same pivot}. Hence, the above loop runs for no more than $O(\fcnt)$ iterations.
See \Cref{fig:cascade} for an illustration.

\begin{figure}[!tb]
  \centering
  \opt{SoCG}{\includegraphics[width=\linewidth]{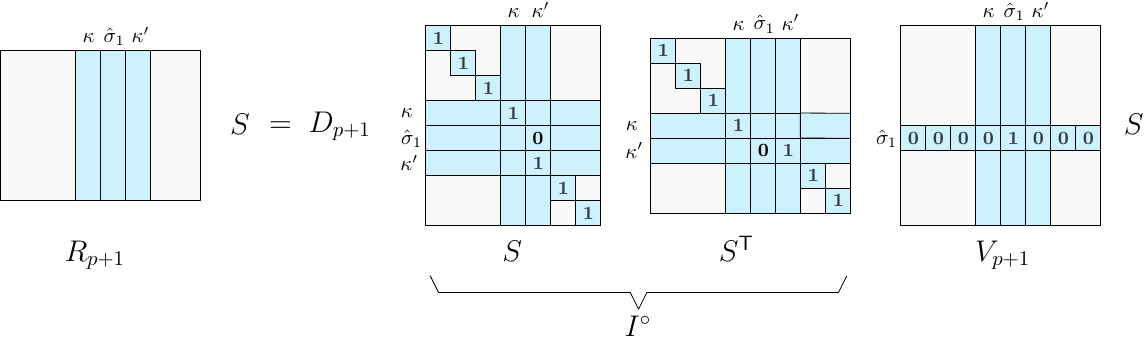}}
  \opt{arXiv,siam}{\includegraphics[width=.9\linewidth]{fig/RT-DTTVT}}
  \caption{Illustration for the equation $\Rmat_{\Dim+1}\Tmat=\Dmat_{\Dim+1}\Tmat\Tmat^\tran\Vmat_{\Dim+1}\Tmat$.}
  \label{fig:RT-DTTVT}
\end{figure}

\begin{figure}[!tb]
  \centering
  \includegraphics[width=.8\linewidth]{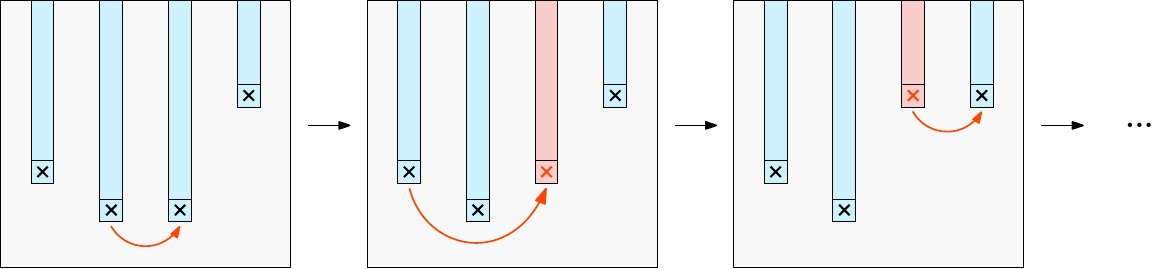} 
  \caption{Illustration of the cascade sum where pivot row index of the updated column always becomes smaller.}
  \label{fig:cascade}
\end{figure}

\para{Summary.}
We summarize the update procedure
in~\Cref{sec:proc-out-contra}.
Notice that we do not need to actually perform
the matrix multiplications as in 
\Cref{eqn:rvmat-mult} to get $\Rmatx_{\Dim+1}$ and $\Vmatx_{\Dim+1}$:
we can simply perform the row addition and row/column removal as needed.
The running time is then dominated by the column additions performed.
Since there are no more than $O(\fcnt)$ column 
additions, the time complexity is
$O(\fcnt^2)$.

\subsection{Update for degree $\Dim$ and $\Dim+2$\opt{siam}{.}}\label{sec:contra-up-other}
By \Cref{obsv:contra-adj-change},
\[\Dmatx_{\Dim+2}=\Smat_{\Dim+2}\Dmat_{\Dim+2},\] 
where $\Smat_{\Dim+2}$ is
similar to the matrix $\Smat$ in \Cref{sec:contra-up}.
Let 
\[\Rmatx_{\Dim+2}=\Smat_{\Dim+2}\Rmat_{\Dim+2}
\text{ and }\Vmatx_{\Dim+2}=\Vmat_{\Dim+2}.\]
Then, \[\Rmatx_{\Dim+2}=\Dmatx_{\Dim+2}\Vmatx_{\Dim+2}.\]
Notice that there may be two columns in $\Rmatx_{\Dim+2}$
having the same pivot, and we perform the cascade sum on $\Rmatx_{\Dim+2}$ and $\Vmatx_{\Dim+2}$ as in
\Cref{sec:contra-up} to make the pivots distinct.

To see the update for degree $\Dim$,
first notice that $\partial\ssx_1=\partial\ssx_2$
because
there are no insertions and deletions
in between the two arrows $\cplx_{i-1}\bakinctosp{\ssx} 
\cplx_i 
\inctosp{\ssx} \cplx_{i+1}$ in $\filt$.
Then, $\ssx_2$ is a \emph{positive} cell in $\filtnz$,
where $\ssx_1+\ssx_2$ is a cycle created by $\ssx_2$.
We then have 
$\col{\Rmat_{\Dim}}{{\ssx}_2}=0$. So we only need to
add $\col{\Vmat_{\Dim}}{{\ssx}_2}$ to any other $\col{\Vmat_{\Dim}}{\xssx}$
with $\matel{\Vmat_{\Dim}}{{\ssx}_2}{\xssx}\neq 0$,
as done similarly in Step I in \Cref{sec:contra-up}.
Moreover, we have 
\[\Dmatx_{\Dim}=\Dmat_{\Dim}\Tmat_{\Dim},\]
where $\Tmat_{\Dim}$ is 
similar to the matrix $\Tmat$ in \Cref{sec:contra-up}.
Let
\[\Rmatx_{\Dim}=\Rmat_{\Dim}\Tmat_{\Dim}\text{ and }
\Vmatx_{\Dim}=\Tmat_{\Dim}^\tran\Vmat_{\Dim}\Tmat_{\Dim}.\]
We then have
\begin{equation*}
\Rmatx_{\Dim}=\Dmatx_{\Dim}\Vmatx_{\Dim}, 
\end{equation*}
whose justification is the same as that for \Cref{eqn:RT-DTTVT}
given by the proof of \Cref{prop:Rx-DxVx}.
Notice that columns in $\Rmatx_{\Dim}$ have
distinct pivots.

Since the update presented above 
is also dominated by the column additions
and there are no more than $O(\fcnt)$ of them,
the update
for degree $\Dim$ and $\Dim+2$ also takes $O(\fcnt^2)$ time.

We now conclude:

\begin{theorem}\label{thm:out-contra-complexity}
For an outward contraction operation,
given a reduced decomposition 
$\Rmat_{\dimx}=\Dmat_{\dimx}\Vmat_{\dimx}$
for $\filtnz$ for each $\dimx$,
one can obtain a reduced decomposition 
$\Rmatx_{\dimx}=\Dmatx_{\dimx}\Vmatx_{\dimx}$
for $\filtxnz$ for each $\dimx$ in $O(\fcnt^2)$ time.
\end{theorem}
\begin{proof}
We have justified the time complexity of the algorithms presented. The correctness of
the update follows from the fact that the decomposition
$\Rmatx_{\dimx}=\Dmatx_{\dimx}\Vmatx_{\dimx}$ we maintain 
is always valid at each step (e.g., we only perform left-to-right column additions)
and that it is reduced at the end.
\end{proof}

\section{Outward expansion\opt{siam}{.}}\label{sec:out-exp}

We present in this section the update algorithm for outward expansion,
which is the inverse operation of outward contraction.
Throughout the section, as in \Cref{eqn:out-contrac-in-sec}, let $\filtx$
 be the filtration \emph{before}
the expansion,
and $\filt$,
the filtration after.
We have that the change from $\Dmatx_\dimx$ to $\Dmat_\dimx$
for each $\dimx$ is also the inverse of the change 
from $\Dmat_\dimx$ to $\Dmatx_\dimx$,
as in \Cref{obsv:contra-adj-change}, e.g.,
we still have $\Dmatx_{\Dim+1}=\Smat\Dmat_{\Dim+1}\Tmat$.

\subsection{Update for degree $\Dim+1$\opt{siam}{.}}\label{sec:exp-up}

We describe in detail the update for degree $\Dim+1$, 
i.e., given a reduced decomposition
\begin{equation}\label{eqn:exp-given-pplus1}
\Rmatx_{\Dim+1}=\Dmatx_{\Dim+1}\Vmatx_{\Dim+1}
\end{equation}
for $\filtxnz$,
we find the corresponding decomposition for $\filtnz$.
The updates for the degrees $\Dim$, $\Dim+2$ are briefly described in \Cref{sec:exp-up-other}.

Notice that the update for outward expansion
is not a simple ``inverse'' of the update for outward contraction.
For example, 
while adding rows $\row{\Dmat_{\Dim+1}}{\ssx_1}$, $\row{\Dmat_{\Dim+1}}{\ssx_2}$ 
to form the row $\row{\Dmatx_{\Dim+1}}{\ssx_0}$
can be written as a left multiplication by $\Smat$,
we cannot do so for the
``splitting'' of non-zero entries in $\row{\Dmatx_{\Dim+1}}{\ssx_0}$
to get the rows $\row{\Dmat_{\Dim+1}}{\ssx_1}$, $\row{\Dmat_{\Dim+1}}{\ssx_2}$.

To perform the update, first let 
$\Dmat'_{\Dim+1}$ be the matrix derived
from $\Dmatx_{\Dim+1}$ by ``splitting'' the non-zero entries
in $\row{\Dmatx_{\Dim+1}}{\ssx_0}$ to form the rows 
of $\ssx_1$, $\ssx_2$,
without adding the column of $\cdec{\ssx}_1$.
In other words, 
$\Dmat'_{\Dim+1}:=\Dmat_{\Dim+1}\Tmat$;
see \Cref{fig:mat_adj_chg}.
Now,
suppose that we are given the following reduced decomposition
\begin{equation}\label{sec:decomp-prime}
\Rmat'_{\Dim+1}=\Dmat'_{\Dim+1}\Vmat'_{\Dim+1},
\end{equation}
where $\Vmat'_{\Dim+1}$ has the same dimension as $\Vmatx_{\Dim+1}$.
Then, we can easily obtain a reduced decomposition
$\Rmat_{\Dim+1}=\Dmat_{\Dim+1}\Vmat_{\Dim+1}$
in $O(\fcnt^2)$ time
as follows: (1) 
let $\Vmat_{\Dim+1}$ be derived from $\Vmat'_{\Dim+1}$
by appending a new column that represents
a chain containing the single cell $\cdec{\ssx}_1$, 
and let $\Rmat_{\Dim+1}$ be derived from $\Rmat'_{\Dim+1}$
by appending $\partial\cdec{\ssx}_1$;
(2) perform a single iteration of the left-to-right persistence reduction~\cite{edelsbrunner2000topological} 
to make $\Rmat_{\Dim+1}$ reduced;
(3) using the algorithm in~\cite{cohen2006vines}, perform $O(\fcnt)$ transpositions to place the columns of 
$\cdec{\ssx}_1$ in the right place in $\Rmat_{\Dim+1}$ and $\Vmat_{\Dim+1}$,
with each transposition taking $O(\fcnt)$ time.

The rest of the subsection focuses on obtaining a
reduced decomposition as in \Cref{sec:decomp-prime}
from the one in~\Cref{eqn:exp-given-pplus1}.
Since columns in $\Dmatx_{\Dim+1}$ and $\Dmat'_{\Dim+1}$
have a one-to-one correspondence, 
let $\Vmat'_{\Dim+1}=\Vmatx_{\Dim+1}$.
We ``recompute'' the boundaries of columns in $\Vmat'_{\Dim+1}$
by letting $\Rmat'_{\Dim+1}=\Dmat'_{\Dim+1}\Vmat'_{\Dim+1}$.
We then have:
\begin{proposition}\label{obsv:col-change}
For each $(\Dim+1)$-cell $\xssx$ in $\filtxnz$,
one of the following holds concerning the change from
the column $\col{\Rmatx_{\Dim+1}}{\xssx}$
to the column $\col{\Rmat'_{\Dim+1}}{\xssx}${\rm:}

\begin{itemize}
    \item Case 1: $\matel{\Rmatx_{\Dim+1}}{\ssx_0}{\xssx}=0$; $\matel{\Rmat'_{\Dim+1}}{\ssx_1}{\xssx}=0$ and 
    $\matel{\Rmat'_{\Dim+1}}{\ssx_2}{\xssx}=0$.
    \item Case 2: $\matel{\Rmatx_{\Dim+1}}{\ssx_0}{\xssx}=0$; $\matel{\Rmat'_{\Dim+1}}{\ssx_1}{\xssx}=1$ and 
    $\matel{\Rmat'_{\Dim+1}}{\ssx_2}{\xssx}=1$.
    \item Case 3: $\matel{\Rmatx_{\Dim+1}}{\ssx_0}{\xssx}=1$; $\matel{\Rmat'_{\Dim+1}}{\ssx_1}{\xssx}=1$ and 
    $\matel{\Rmat'_{\Dim+1}}{\ssx_2}{\xssx}=0$.
    \item Case 4: $\matel{\Rmatx_{\Dim+1}}{\ssx_0}{\xssx}=1$; $\matel{\Rmat'_{\Dim+1}}{\ssx_1}{\xssx}=0$ and 
    $\matel{\Rmat'_{\Dim+1}}{\ssx_2}{\xssx}=1$.
\end{itemize}

\noindent
All other entries in the two columns are the same.
\end{proposition}

\begin{proof}
Suppose that $\matel{\Rmatx_{\Dim+1}}{\ssx_0}{\xssx}=0$. Then, $\ssx_0$ is a face of
even $(\Dim+1)$-cells in the chain 
$\col{\Vmatx_{\Dim+1}}{\xssx}$.
Notice that these $(\Dim+1)$-cells 
either have $\ssx_1$ or $\ssx_2$
as a $\Dim$-face in $\filtnz$.
Then, either both $\ssx_1$, $\ssx_2$ are faces of 
even $(\Dim+1)$-cells in
$\col{\Vmat'_{\Dim+1}}{\xssx}=\col{\Vmatx_{\Dim+1}}{\xssx}$ (Case 1)
or both are faces of 
odd $(\Dim+1)$-cells in
$\col{\Vmat'_{\Dim+1}}{\xssx}$ (Case 2).
It is also not hard to see Case 3 and 4
when $\ssx_0$ is a face of
odd $(\Dim+1)$-cells in
$\col{\Vmatx_{\Dim+1}}{\xssx}$.
\end{proof}

\begin{proposition}\label{prop:col-piv}
For each $(\Dim+1)$-cell $\xssx$ in $\filtxnz$,
$\pivot(\col{\Rmatx_{\Dim+1}}{\xssx})\neq\pivot(\col{\Rmat'_{\Dim+1}}{\xssx})$
if and only if $\pivot(\col{\Rmat'_{\Dim+1}}{\xssx})=\ssx_2$.
\end{proposition}
\begin{proof}
Suppose that $\pivot(\col{\Rmatx_{\Dim+1}}{\xssx})\neq\pivot(\col{\Rmat'_{\Dim+1}}{\xssx})$.
Then, either Case 2 or Case 4 in \Cref{obsv:col-change} happens because
the other cases cannot make the pivot change.
If $\pivot(\col{\Rmatx_{\Dim+1}}{\xssx})>\ssx_2$
(here we naturally consider the $\Dim$-cell $\pivot(\col{\Rmatx_{\Dim+1}}{\xssx})$ in $\filtxnz$
as a  $\Dim$-cell in $\filtnz$),
then $\pivot(\col{\Rmatx_{\Dim+1}}{\xssx})=\pivot(\col{\Rmat'_{\Dim+1}}{\xssx})$.
So we have $\pivot(\col{\Rmatx_{\Dim+1}}{\xssx})<\ssx_2$,
which implies that $\pivot(\col{\Rmat'_{\Dim+1}}{\xssx})=\ssx_2$.

Now suppose that $\pivot(\col{\Rmat'_{\Dim+1}}{\xssx})=\ssx_2$.
Then, $\pivot(\col{\Rmatx_{\Dim+1}}{\xssx})\neq\pivot(\col{\Rmat'_{\Dim+1}}{\xssx})$
since $\pivot(\col{\Rmatx_{\Dim+1}}{\xssx})\neq\ssx_2$.
\end{proof}

\begin{proposition}\label{prop:recov-piv}
For any  
$(\Dim+1)$-cells $\xssx$, $\xssxx$ in $\filtxnz$ with
$\pivot(\col{\Rmat'_{\Dim+1}}{\xssx})=\pivot(\col{\Rmat'_{\Dim+1}}{\xssxx})=\ssx_2$,
one has 
$\pivot(\col{\Rmatx_{\Dim+1}}{\xssx}+\col{\Rmatx_{\Dim+1}}{\xssxx})=\pivot(\col{\Rmat'_{\Dim+1}}{\xssx}+\col{\Rmat'_{\Dim+1}}{\xssxx})$.
\end{proposition}
\begin{proof}
Since $\pivot(\col{\Rmat'_{\Dim+1}}{\xssx})=\pivot(\col{\Rmat'_{\Dim+1}}{\xssxx})=\ssx_2$,
the columns of $\xssx$ (resp.\ $\xssxx$)
in $\Rmatx_{\Dim+1}$ and $\Rmat'_{\Dim+1}$
fall in either Case 2 or  4 in \Cref{obsv:col-change}.
Denote $\col{\Rmatx_{\Dim+1}}{\xssx}+\col{\Rmatx_{\Dim+1}}{\xssxx}$
as $\mathrm{col}_1$ and $\col{\Rmat'_{\Dim+1}}{\xssx}+\col{\Rmat'_{\Dim+1}}{\xssxx}$
as $\mathrm{col}_2$.
We then have $\mathrm{col}_1[\ssx_0]=\mathrm{col}_2[\ssx_1]$,
which can be verified by a case analysis of the different values of
$\matel{\Rmatx_{\Dim+1}}{\ssx_0}{\xssx}$, $\matel{\Rmatx_{\Dim+1}}{\ssx_0}{\xssxx}$, 
$\matel{\Rmat'_{\Dim+1}}{\ssx_1}{\xssx}$, and $\matel{\Rmat'_{\Dim+1}}{\ssx_1}{\xssxx}$.
We omit the case analysis.
Since $\mathrm{col}_2[\ssx_2]=0$ and the other entries 
of $\mathrm{col}_1$, $\mathrm{col}_2$ are the same, we have 
$\pivot(\mathrm{col}_1)=\pivot(\mathrm{col}_1)$.
\end{proof}

\Cref{prop:col-piv} characterizes those columns whose pivots change from 
$\Rmatx_{\Dim+1}$ to $\Rmat'_{\Dim+1}$.
\Cref{prop:recov-piv} indicates that summing two such columns in 
$\Rmat'_{\Dim+1}$ ``recovers'' the pivot of the sum in $\Rmatx_{\Dim+1}$.
Since $\Rmatx_{\Dim+1}$ have distinct pivots,
we utilize the above facts to guide our update to make $\Rmat'_{\Dim+1}$
reduced.
We then have the following steps.

\begin{itemize}

    \item Step 1:
    While there are two columns $\col{\Rmat'_{\Dim+1}}{\xssx}$ and $\col{\Rmat'_{\Dim+1}}{\xssxx}$,
    $\xssx<\xssxx$, with pivots being $\ssx_2$
    such that 
    \[\pivot(\col{\Rmatx_{\Dim+1}}{\xssx})<\pivot(\col{\Rmatx_{\Dim+1}}{\xssxx}),\]
    add $\col{\Rmat'_{\Dim+1}}{\xssx}$ to $\col{\Rmat'_{\Dim+1}}{\xssxx}$
    and $\col{\Vmat'_{\Dim+1}}{\xssx}$ to $\col{\Vmat'_{\Dim+1}}{\xssxx}$.
    After this, $\col{\Rmat'_{\Dim+1}}{\xssxx}$ ``regains'' its old (unique) pivot
    in $\col{\Rmatx_{\Dim+1}}{\xssxx}$,
    i.e., $\pivot(\col{\Rmat'_{\Dim+1}}{\xssxx})=\pivot(\col{\Rmatx_{\Dim+1}}{\xssxx})$,
    due to \Cref{prop:col-piv} and the fact that
    $\pivot(\col{\Rmatx_{\Dim+1}}{\xssx})<\pivot(\col{\Rmatx_{\Dim+1}}{\xssxx})$
    
    \item Step 2:
    Let
    \[\col{\Rmat'_{\Dim+1}}{\xssx_1},\ldots,\col{\Rmat'_{\Dim+1}}{\xssx_\ell}\]
    be all columns in $\Rmat'_{\Dim+1}$
    whose pivots are still $\ssx_2$,
    where \[\xssx_1<\cdots<\xssx_\ell.\]
    We have \[\pivot(\col{\Rmatx_{\Dim+1}}{\xssx_1})>\cdots>\pivot(\col{\Rmatx_{\Dim+1}}{\xssx_\ell}),\]
    and specifically
    $\xssx_1,\ldots,\xssx_\ell$ are at the {Pareto front}
    of the partial order defined by the order 
    of column indices and the order of pivots in $\Rmatx_{\Dim+1}$. 
    
    We then perform the following update:
    \begin{itemize}
        \item For $j=\ell,\ell-1,\ldots,2$,
        add $\col{\Rmat'_{\Dim+1}}{\xssx_{j-1}}$
        to $\col{\Rmat'_{\Dim+1}}{\xssx_j}$ and
        $\col{\Vmat'_{\Dim+1}}{\xssx_{j-1}}$ to
        $\col{\Vmat'_{\Dim+1}}{\xssx_j}$.
    \end{itemize}

    \opt{SoCG}{\smallskip}
    After this, each $\col{\Rmat'_{\Dim+1}}{\xssx_{j}}$ with $j>1$
    gets the pivot of $\col{\Rmatx_{\Dim+1}}{\xssx_{j-1}}$
    because $\pivot(\col{\Rmatx_{\Dim+1}}{\xssx_{j-1}})>\pivot(\col{\Rmatx_{\Dim+1}}{\xssx_{j}})$.
    Also, $\col{\Rmat'_{\Dim+1}}{\xssx_{1}}$
    has $\ssx_2$ as pivot. Hence, columns in $\Rmat'_{\Dim+1}$ now have
    distinct pivots.
\end{itemize}

\para{Summary.}
We summarize the update procedure
in~\Cref{sec:proc-out-expan}.
Since no more than $O(\fcnt)$ column 
additions are performed, the time complexity is
$O(\fcnt^2)$.

\subsection{Update for degree $\Dim$ and $\Dim+2$\opt{siam}{.}}\label{sec:exp-up-other}

To obtain $\Rmat_{\Dim}=\Dmat_{\Dim}\Vmat_{\Dim}$, one way is to use the approach described
at the beginning of \Cref{sec:exp-up}, that is, to append $\partial\ssx_2$ and $\ssx_2$
to $\Rmatx_{\Dim}$ and $\Vmatx_{\Dim}$, respectively, and to perform $O(\fcnt)$ transpositions
to place the two new columns in the correct positions.
However, this is an overkill: since $\partial\ssx_1=\partial\ssx_2$ (see \Cref{sec:contra-up-other}),
inserting $\ssx_2$ in $\filtnz$ always creates a new cycle
$\ssx_1+\ssx_2$. Hence, we only need to add a column $0$ and a column
$\ssx_1+\ssx_2$ to $\Rmatx_{\Dim}$ and $\Vmatx_{\Dim}$, respectively,
at the appropriate place for $\ssx_2$, which finishes the update for degree $\Dim$.

For degree $\Dim+2$, we notice that the change from $\Dmatx_{\Dim+2}$ to $\Dmat_{\Dim+2}$
is just the ``splitting'' of $\row{\Dmatx_{\Dim+2}}{\cdec{\ssx}_0}$
to get $\row{\Dmat_{\Dim+2}}{\cdec{\ssx}_2}$ and $\row{\Dmat_{\Dim+2}}{\cdec{\ssx}_1}$.
For this, we use the update procedure for obtaining the decomposition in \Cref{sec:decomp-prime}
from that in~\Cref{eqn:exp-given-pplus1}, as described in \Cref{sec:exp-up}.

We now conclude:

\begin{theorem}
For an outward expansion operation,
given a reduced decomposition 
$\Rmatx_{\dimx}=\Dmatx_{\dimx}\Vmatx_{\dimx}$
for $\filtxnz$ for each $\dimx$,
one can obtain a reduced decomposition 
$\Rmat_{\dimx}=\Dmat_{\dimx}\Vmat_{\dimx}$
for $\filtnz$ for each $\dimx$ in $O(\fcnt^2)$ time.
\label{thm:outexpand}
\end{theorem}

 Combining Theorems~\ref{thm:extr-zzrep}, \ref{thm:out-contra-complexity}, \ref{thm:outexpand} and our discussion about easy updates, we get the
final result.

\begin{theorem}
   Given a simplex-wise zigzag filtration $\filt$ with $m$ insertions and deletions of simplices, 
   one can update the corresponding $R=DV$ decomposition in $O(m)$ time for switches and in $O(m^2)$ time for contractions and expansions.  One can extract the zigzag representative in $O(\log m +C)$ time for each interval, where $C$ is the size of the representative.
\end{theorem}

\section{Conclusion\opt{siam}{.}} 

In this paper, we propose algorithms for 
updating zigzag representatives over eight operations
making  local changes to a zigzag filtration.
The update is performed
on the $R=DV$ decomposition
for a non-zigzag filtration constructed from the original zigzag filtration.
The zigzag representatives can then be obtained based on \Cref{thm:extr-zzrep}.
Complexities for
updating the $R=DV$ decomposition over the operations are summarized as follows:
\begin{itemize}
    \item Forward/backward switch: $O(\fcnt)$;
    \item Inward/outward switch: $O(1)$;
    \item Inward/outward expansion/contraction: $O(\fcnt^2)$.
\end{itemize}

The $R=DV$ decomposition for persistence is not unique~\cite{cohen2006vines}.
Just like~\cite{cohen2006vines,giunti2023pruning,luo2024accelerating},
our updates may not produce the same decomposition
as the standard reduction~\cite{edelsbrunner2000topological},
and therefore the same representatives.
Also,
our updates may not preserve the representatives
after an update operation and its reversal,
just like~\cite{cohen2006vines}.
It is an  interesting open question
whether there are efficient algorithms for such 
updates.

\section*{Acknowledgements\opt{siam}{.}}
TD was supported by NSF funds DMS-2301360 and CCF-2437030.
TH was supported by NSF fund CCF 2439255.
DM was supported by the U.S. Department of Energy, Office of Science,
Office of Advanced Scientific Computing Research's Applied Mathematics
Competitive Portfolios program under Contract No. DE-AC02-05CH11231.

\bibliography{refs}

\appendix

\section{Summary of update procedure for degree $\Dim+1$ for outward contraction\opt{siam}{.}}\label{sec:proc-out-contra}
Given a reduced decomposition 
\[\Rmat_{\Dim+1}=\Dmat_{\Dim+1}\Vmat_{\Dim+1}\]
for $\filtnz$, do the following:

\begin{enumerate}
    \item If $\col{\Rmat_{\Dim+1}}{\cdec{\ssx}_1}=0$:
add $\col{\Vmat_{\Dim+1}}{\cdec{\ssx}_1}$ to any other $\col{\Vmat_{\Dim+1}}{\xssx}$
with $\matel{\Vmat_{\Dim+1}}{\cdec{\ssx}_1}{\xssx}\neq 0$.
    \item 
    If $\col{\Rmat_{\Dim+1}}{\cdec{\ssx}_1}\neq 0$:
    \begin{enumerate}
        \item While there are two columns $\col{\Vmat_{\Dim+1}}{\xssx}$ and $\col{\Vmat_{\Dim+1}}{\xssxx}$,
        $\xssx<\xssxx$, with
        non-zero entries on $\cdec{\ssx}_1$ such that 
        \[\pivot(\col{\Rmat_{\Dim+1}}{\xssx})<\pivot(\col{\Rmat_{\Dim+1}}{\xssxx}),\]
        add $\col{\Vmat_{\Dim+1}}{\xssx}$ to $\col{\Vmat_{\Dim+1}}{\xssxx}$
        and $\col{\Rmat_{\Dim+1}}{\xssx}$ to $\col{\Rmat_{\Dim+1}}{\xssxx}$.
    
        \item Let
        \[\col{\Vmat_{\Dim+1}}{\xssx_1},\ldots,\col{\Vmat_{\Dim+1}}{\xssx_\ell}\]
        be all columns in $\Vmat_{\Dim+1}$
        that still have non-zero entries on $\cdec{\ssx}_1$,
        where \[\xssx_1=\cdec{\ssx}_1<\xssx_2<\cdots<\xssx_\ell.\]
        \item For $j=\ell,\ell-1,\ldots,2$,
            add $\col{\Rmat_{\Dim+1}}{\xssx_{j-1}}$
            to $\col{\Rmat_{\Dim+1}}{\xssx_j}$ and
            $\col{\Vmat_{\Dim+1}}{\xssx_{j-1}}$ to
            $\col{\Vmat_{\Dim+1}}{\xssx_j}$.
    \end{enumerate}

    \item Let 
\begin{equation*}
\Rmatx_{\Dim+1}=\Smat\Rmat_{\Dim+1}\Tmat\text{ and }\Vmatx_{\Dim+1}=\Tmat^\tran\Vmat_{\Dim+1}\Tmat.
\end{equation*}
    \item 
    While there are two columns $\col{\Rmatx_{\Dim+1}}{\xssx}$ and $\col{\Rmatx_{\Dim+1}}{\xssxx}$,
        $\xssx<\xssxx$, with the same pivot,
        add $\col{\Rmatx_{\Dim+1}}{\xssx}$ to $\col{\Rmatx_{\Dim+1}}{\xssxx}$
        and $\col{\Vmatx_{\Dim+1}}{\xssx}$ to $\col{\Vmatx_{\Dim+1}}{\xssxx}$.
\end{enumerate}

\section{Summary of update procedure for degree $\Dim+1$ for outward expansion\opt{siam}{.}}\label{sec:proc-out-expan}

Given a reduced decomposition
\begin{equation*}
\Rmatx_{\Dim+1}=\Dmatx_{\Dim+1}\Vmatx_{\Dim+1}
\end{equation*}
for $\filtxnz$,
do the following:

\begin{enumerate}
  \item
Let $\Vmat'_{\Dim+1}=\Vmatx_{\Dim+1}$ and $\Rmat'_{\Dim+1}=\Dmat'_{\Dim+1}\Vmat'_{\Dim+1}$.

\item
While there are two columns $\col{\Rmat'_{\Dim+1}}{\xssx}$ and $\col{\Rmat'_{\Dim+1}}{\xssxx}$,
    $\xssx<\xssxx$, with pivots being $\ssx_2$
    such that 
    \[\pivot(\col{\Rmatx_{\Dim+1}}{\xssx})<\pivot(\col{\Rmatx_{\Dim+1}}{\xssxx}),\]
    add $\col{\Rmat'_{\Dim+1}}{\xssx}$ to $\col{\Rmat'_{\Dim+1}}{\xssxx}$
    and $\col{\Vmat'_{\Dim+1}}{\xssx}$ to $\col{\Vmat'_{\Dim+1}}{\xssxx}$.

    \item
      Let
    $\col{\Rmat'_{\Dim+1}}{\xssx_1},\ldots,\col{\Rmat'_{\Dim+1}}{\xssx_\ell}$
    be all columns in $\Rmat'_{\Dim+1}$
    whose pivots are still $\ssx_2$,
    where $\xssx_1<\cdots<\xssx_\ell$.
    
    For $j=\ell,\ell-1,\ldots,2$,
        add $\col{\Rmat'_{\Dim+1}}{\xssx_{j-1}}$
        to $\col{\Rmat'_{\Dim+1}}{\xssx_j}$ and
        $\col{\Vmat'_{\Dim+1}}{\xssx_{j-1}}$ to
        $\col{\Vmat'_{\Dim+1}}{\xssx_j}$.

\item
Let $\Vmat_{\Dim+1}$ be derived from $\Vmat'_{\Dim+1}$
by appending a new column that represents
a chain containing the single cell $\cdec{\ssx}_1$, 
and let $\Rmat_{\Dim+1}$ be derived from $\Rmat'_{\Dim+1}$
by appending $\partial\cdec{\ssx}_1$.
\item Perform a single iteration of the left-to-right persistence reduction~\cite{edelsbrunner2000topological} 
to make $\Rmat_{\Dim+1}$ reduced.
\item Using the algorithm in~\cite{cohen2006vines}, perform $O(\fcnt)$ transpositions to place the columns of 
$\cdec{\ssx}_1$ in the right place in $\Rmat_{\Dim+1}$ and $\Vmat_{\Dim+1}$.

\end{enumerate}

\section{Generalization to arbitrary field coefficients\opt{siam}{.}}\label{sec:gen-field}

We describe the generalization of our update algorithm to coefficients in an arbitrary field $\field$.
The update for outward contraction in \Cref{sec:out-contra}
generalizes immediately because it is expressed entirely as matrix multiplications. So 
we only describe here how to generalize the update for outward expansion,
which boils down to adapting
\Cref{obsv:col-change,prop:recov-piv}.
With a general $\field$,
each cell needs an orientation.
While orientations of most cells can be assigned arbitrarily,
we have one restriction: since the $\DG$-cells $\ssx_1$,
$\ssx_2$ have the same set of vertices, let $\ssx_1$,
$\ssx_2$ have an orientation induced by the same ordering of vertices.
We then have:
\begin{proposition}[Adaption of \Cref{obsv:col-change} to arbitrary $\field$]
\label{obsv:col-change-adp}
Let $\xssx$ be any $(\Dim+1)$-cell in $\filtxnz$.
Then,
$\matel{\Rmatx_{\Dim+1}}{\ssx_0}{\xssx}=\matel{\Rmat'_{\Dim+1}}{\ssx_1}{\xssx}+\matel{\Rmat'_{\Dim+1}}{\ssx_2}{\xssx}$,
and
$\matel{\Rmatx_{\Dim+1}}{\xssxx}{\xssx}=\matel{\Rmat'_{\Dim+1}}{\xssxx}{\xssx}$
for any other $\Dim$-cell $\xssxx$.
\end{proposition}

We omit the proof of \Cref{obsv:col-change-adp}, which follows the same idea
as the proof of \Cref{obsv:col-change}.
Notice that \Cref{prop:col-piv} is still true here.

\begin{proposition}[Adaption of \Cref{prop:recov-piv} to arbitrary $\field$]
\label{prop:recov-piv-adp}
For any  
$(\Dim+1)$-cells $\xssx$, $\xssxx$ in $\filtxnz$ with
$\pivot(\col{\Rmat'_{\Dim+1}}{\xssx})=\pivot(\col{\Rmat'_{\Dim+1}}{\xssxx})=\ssx_2$,
let $\aG_2=\matel{\Rmat'_{\Dim+1}}{\ssx_2}{\xssx}$ and
$\gG_2=\matel{\Rmat'_{\Dim+1}}{\ssx_2}{\xssxx}$.
Then, for $\mu=-\gG_2/\aG_2$, $\pivot(\mu\cdot\col{\Rmatx_{\Dim+1}}{\xssx}+\col{\Rmatx_{\Dim+1}}{\xssxx})=\pivot(\mu\cdot\col{\Rmat'_{\Dim+1}}{\xssx}+\col{\Rmat'_{\Dim+1}}{\xssxx})$.
\end{proposition}
\begin{proof}
Let $\mathrm{col}_1=\mu\cdot\col{\Rmatx_{\Dim+1}}{\xssx}+\col{\Rmatx_{\Dim+1}}{\xssxx}$
 and $\mathrm{col}_2=\mu\cdot\col{\Rmat'_{\Dim+1}}{\xssx}+\col{\Rmat'_{\Dim+1}}{\xssxx}$.
Notice that $\mathrm{col}_2[\ssx_2]=0$ and
$\mathrm{col}_1[\xssxx]=\mathrm{col}_2[\xssxx]$
for any $\xssxx\not\in\{\ssx_0,\ssx_1,\ssx_2\}$.
So we only need to show $\mathrm{col}_1[\ssx_0]=\mathrm{col}_2[\ssx_1]$.
Let
$\aG_0=\matel{\Rmatx_{\Dim+1}}{\ssx_0}{\xssx}$,
$\gG_0=\matel{\Rmatx_{\Dim+1}}{\ssx_0}{\xssxx}$,
$\aG_1=\matel{\Rmat'_{\Dim+1}}{\ssx_1}{\xssx}$, and
$\gG_1=\matel{\Rmat'_{\Dim+1}}{\ssx_1}{\xssxx}$.
By \Cref{obsv:col-change-adp},
we have $\aG_0=\aG_1+\aG_2$ and $\gG_0=\gG_1+\gG_2$.
Then,
\[\mathrm{col}_2[\ssx_1]=(-\gG_2/\aG_2)\aG_1+\gG_1=(1/\aG_2)(-\aG_1\gG_2+\aG_2\gG_1).\]
Also,
\begin{align*}
\mathrm{col}_1[\ssx_0]&=(-\gG_2/\aG_2)\aG_0+\gG_0=(-\gG_2/\aG_2)(\aG_1+\aG_2)+\gG_1+\gG_2\\
&=(1/\aG_2)(-\aG_1\gG_2-\aG_2\gG_2+\aG_2\gG_1+\aG_2\gG_2)=(1/\aG_2)(-\aG_1\gG_2+\aG_2\gG_1)=\mathrm{col}_2[\ssx_1].\opt{arXiv}{\qedhere}
\end{align*}
\end{proof}

We can then perform the update for outward expansion
for degree $\Dim+1$ in \Cref{sec:exp-up}. The generalization to field $\field$
for degree $\Dim$, $\Dim+2$ also follows.

\opt{arXiv}{\end{document}}